\documentclass[aps,pra,onecolumn,nofootinbib,longbibliography,a4paper,superscriptaddress,11pt,tightenlines]{revtex4-2}

\usepackage[utf8]{inputenc}
\usepackage[T1]{fontenc}
\usepackage{amsthm}
\usepackage{amsmath}
\usepackage{amsfonts}
\usepackage{amssymb}
\usepackage{lmodern}
\usepackage{thmtools}
\usepackage{thm-restate}
\usepackage{hyperref}
\usepackage{xcolor}
\usepackage{enumitem}
\usepackage{graphicx}
\usepackage[english]{babel}
\usepackage{bbm}
\usepackage{environ}

\usepackage{tikz}

\newtheorem{theorem}{Theorem}
\newtheorem{corollary}{Corollary}
\newtheorem{lemma}{Lemma}
\newtheorem{proposition}{Proposition}
\newtheorem{definition}{Definition}
\newtheorem{remark}{Remark}

\NewEnviron{myalign}{
	\begin{equation}\begin{split}
			\BODY
	\end{split}\end{equation}
}

\newcommand{\myleft}{\mathopen{}\mathclose\bgroup\left}
\newcommand{\myright}{\aftergroup\egroup\right}

\newcommand{\id}{\ensuremath\mathrm{id}}

\DeclareMathOperator{\Tr}{Tr}

\DeclareMathOperator{\rank}{rank}

\DeclareMathOperator{\Id}{Id}

\DeclareMathOperator{\diag}{diag}

\newcommand{\conv}{\operatorname{conv}}

\DeclareMathOperator{\SO}{SO}
\DeclareMathOperator{\CSP}{CSP}

\DeclareMathOperator{\AD}{AD}

\def\FF{\mathbb{F}}
\def\CC{\mathbb{C}}

\def\stab{\mathrm{STAB}}
\def\SP{\mathrm{SP}}
\DeclareMathOperator{\Cl}{Cl}% Clifford group Cl(n) \subset U(n)
\DeclareMathOperator{\Sp}{Sp}% Clifford group Cl(n) \subset U(n)
\newcommand{\Pa}{\mathcal{P}}

\DeclareMathOperator\tr{Tr}

\DeclareMathOperator\GL{GL}

\newcommand{\ket}[1]{\left.\left|{#1}\right.\right\rangle}
\newcommand{\ketn}[1]{| #1 \rangle}
\newcommand{\ketb}[1]{\bigl| #1 \bigr\rangle}
\newcommand{\bra}[1]{\left.\left\langle{#1}\right.\right|}
\newcommand{\bran}[1]{\langle #1 |}
\newcommand{\brab}[1]{\bigl\langle #1 \bigr|}
\newcommand{\braket}[2]{\left\langle #1 \middle| #2 \right\rangle}

\newcommand{\ketbra}[2]{\ket{#1} \!\! \bra{#2}}

\newcommand{\ketbran}[2]{\ketn{#1} \! \bran{#2}}
\newcommand{\ketbrab}[2]{\ketb{#1} \! \brab{#2}}
\newcommand{\sandwich}[3]
  {\left\langle  #1 \right| #2 \left| #3 \right\rangle}

\newcommand{\R}{\mathbb{R}}
\newcommand{\N}{\mathbb{N}}
\newcommand{\C}{\mathbb{C}}
\newcommand{\Z}{\mathbb{Z}}

\newcommand{\F}{\mathbb{F}}

\newcommand{\ie}{i.\,e.}

\newcommand{\eg}{e.\,g.}
 
\newcommand{\one}{\mathbbm{1}}

\newcommand{\ind}{\mathbf{1}}

\def\be                 {\begin{equation}}
\def\ee                 {\end{equation}}

\setlist[enumerate]{topsep=1ex,itemsep=0.25ex,partopsep=1ex,parsep=1ex}
\setlist[itemize]{topsep=1ex,itemsep=0.25ex,partopsep=1ex,parsep=1ex}

\begin{document}

\title{The axiomatic and the operational approaches to resource theories of magic do not coincide}

\author{Arne Heimendahl}
\thanks{Authors contributed equally.}
\affiliation{Department for Mathematics and Computer Science, University of Cologne, Germany}%
\author{Markus Heinrich}
\thanks{Authors contributed equally.}
\affiliation{Institute for Theoretical Physics, University of Cologne, Germany} 
\affiliation{Institute for Theoretical Physics, Heinrich Heine University Düsseldorf, Germany} 
\author{David Gross}
\affiliation{Institute for Theoretical Physics, University of Cologne, Germany}

\begin{abstract}
    \emph{Stabiliser operations} occupy a prominent role in fault-tolerant quantum computing. 
    They are defined \emph{operationally}: by the use of Clifford gates, Pauli measurements and classical control.
			These operations can be efficiently simulated on a classical computer, a result which is known as the \emph{Gottesman-Knill} theorem.
    However, an additional supply of \emph{magic states} is enough to promote them to a universal, fault-tolerant model for quantum computing.
	To quantify the needed resources in terms of magic states, a \emph{resource theory of magic} has been developed.
    Stabiliser operations (SO) are considered free within this theory, however they are not the most general class of free operations.
    From an axiomatic point of view, these are the \emph{completely stabiliser-preserving} (CSP) channels, defined as those that preserve the convex hull of stabiliser states.
	It has been an open problem to decide whether these two definitions lead to the same class of operations.
    In this work, we answer this question in the negative, by constructing an explicit counter-example.
	This  indicates that recently proposed stabiliser-based simulation techniques of CSP maps are strictly more powerful than Gottesman-Knill-like methods.
	The result is analogous to a well-known fact in entanglement theory, namely that there is a gap between the operationally defined class of local operations and classical communication (LOCC) and the axiomatically defined class of separable channels.
\end{abstract}

\maketitle

%%% =============================================
\section{Introduction}
\label{sec:intro}
%%% =============================================

Despite the advances in the development of quantum platforms, understanding the precise set of quantum phenomena that is required for a quantum advantage over classical computers remains an elusive task.
However, for the design of fault-tolerant quantum computers, it seems imperative to understand these necessary resources.
Here, the \emph{magic state model} of quantum computing offers a particularly fruitful perspective.
In this model, all operations performed by the quantum computer are divided into two classes.
The first class consists of the preparation of stabiliser states, the implementation of Cifford gates, and Pauli measurements.
These \emph{stabiliser operations} by themselves can be efficiently simulated classically by the Gottesman-Knill Theorem \cite{gottesman_stabilizer_1997,aaronson_improved_2004}.
Secondly, the quantum computer needs to be able to prepare \emph{magic states}, defined as states that allow for the implementation of any quantum algorithm when acted on by stabiliser operations \cite{bravyi_universal_2005}. 
In this sense, the magic states provide the ``non-classicality'' required for a quantum advantage.

During recent years, there has been an increasing interest in developing a resource theory of quantum computing that allows for a precise quantification of \emph{magic}.
First resource theories were developed for the somewhat simpler case of odd-dimensional systems, based on a phase-space representation via Wigner functions.
There, the total negativity in the Wigner function of a state is a \emph{resource monotone} called \emph{mana}, and
non-zero mana is a necessary condition for a quantum speed-up \cite{galvao_discrete_2005,gross_non-negative_2007,veitch_negative_2012,veitch_resource_2014,mari_positive_2012,howard_contextuality_2014,delfosse_equivalence_2017}.
In the practically more relevant case of qubits, this theory breaks down, which has led to a number of parallel developments \cite{howard_application_2017,heinrich_robustness_2019,seddon_quantifying_2019,raussendorf_2019,seddon_quantifying_2021,beverland_lower_2020,heimendahl_stabilizer_2021,liu_many-body_2020}.
A common element is that the finite set of stabiliser states, or more generally their convex hull, the \emph{stabiliser polytope}, is taken as the set of free states.
Since stabiliser operations preserve the stabiliser polytope, they are considered free operations in this theory and any monotones should be non-increasing under those.
A number of such \emph{magic monotones} have been studied and their values linked to the runtime of classical simulation algorithms \cite{pashayan_estimating_2015,bravyi_improved_2016,bravyi_simulation_2019,seddon_quantifying_2021}.
In this sense, the degree of magic present in a quantum circuit does seem to correlate with the quantum advantages it confers -- thus validating the premise of the approach.

The set of stabiliser operations ($\SO$) are defined in terms of concrete actions (``prepare a stabiliser state, perform a Clifford unitary, make a measurement, ...'') and thus represent an \emph{operational} approach to defining free transformations in a resource theory of magic.
It is often fruitful to start from an \emph{axiomatic} point of view, by defining the set of free transformations as those physical maps that preserve the set of free states.
This approach has been introduced recently by \textcite{seddon_quantifying_2019}. 
They suggest to refer to a linear map as \emph{completely stabiliser-preserving} (CSP) if it preserves the stabiliser polytope, even when acting on parts of an entangled system.
It has been shown that the magic monotones mentioned above are also non-increasing under CSP maps \cite{seddon_quantifying_2021}.

A natural question is therefore whether the two approaches coincide -- i.e.\ whether $\SO = \CSP$, or whether there are CSP maps that cannot be realised as stabiliser operations \cite{seddon_quantifying_2019}.

To build an intuition for the question, consider the analogous problem in entanglement theory, where the free resources are the separable states.
The axiomatically defined free transformations are the \emph{separable maps} -- completely positive maps that preserve the set of separable states.
The operationally defined free transformations are those that can be realised by local operations and classical communication (LOCC). 
It is known that the set of separable maps is strictly larger than the set of LOCC \cite{bennett_quantum_1999,chitambar_everything_2014} -- a fact that 
leads e.g.\ to a notable gap in the success probability of quantum state discrimination \cite{koashi_quantum_2007,duan_distinguishability_2009} and entanglement conversion \cite{chitambar_increasing_2012} between the two classes.

In this work, we show that -- also in resource theories of magic -- the axiomatic and the operational approaches lead to different classes, that is $\SO \neq \CSP$.

As an auxiliary result, we derive a normal form for stabiliser operations which is used to prove our main result.
From this form, it is evident that any stabiliser operation can be realised in a finite number of rounds -- a statement which is known to not hold for LOCC operations in entanglement theory \cite{chitambar_local_2011}.
Furthermore, we give a characterisation of CSP channels in terms of certain generalised stabiliser measurements and adaptive Clifford operations.
This characterisation has been used in a classical simulation algorithm of CSP channels by \textcite{seddon_quantifying_2021}.

\subsection*{Outline}

In Section~\ref{sec:pre}, we give an introduction to the relevant concepts used throughout the main part of this work.
Next, we prove a minimal version of our main result and illustrate our proof technique for the $2$-qubit case in Section~\ref{sec:min-result}. 
There, we show that there is a $ 2 $-qubit $ \CSP $ channel that is not a stabiliser operation. 
In Section~\ref{sec:gen-results}, we generalise this minimal result to an arbitrary number of qudits. Furthermore, we prove equality of $ \CSP $ and $ \SO $ for a single qudit. 
In Section~\ref{sec:additional}, we describe additional properties of $ \CSP $ channels and give some examples. 
We conclude the main part by commenting on potential implications and future work in Sec.~\ref{sec:summary}.

%%% =============================================
\section{Preliminaries}
\label{sec:pre}
%%% =============================================

% -------------------------------
\subsection{Stabiliser formalism}
\label{sec:stabiliser-formalism}
% -------------------------------

Consider the Hilbert space $\mathcal H = (\C^d)^{\otimes n}$ of $n$ qudits of dimension $d$, where we assume that $d$ is prime.
We label the computational basis $\ket{x}$ by vectors $x$ in the discrete vector space $\F_d^n$. 
Here, $\F_d$ is the finite field of $d$ elements which can be taken to be the residue field $\Z/d\Z$ of integers modulo $d$.
Let $\omega=e^{2\pi i/d}$ be a $d$-th root of unity, then we define the $n$-qudit $Z$ and $X$ operator as usual by their action on the computational basis:
\begin{equation}
\label{eq:pauli-zx}
 Z(z)\ket{y} := \omega^{z\cdot y}\ket{y}, \quad X(x)\ket{y}:=\ket{y+x}, \quad z,x,y\in\F_d^n.
\end{equation}
Here, all operations take place in the finite field $\F_d$ (i.e.~modulo $d$), if not stated otherwise.
To treat the slightly different theory for even and odd $d$ on the same footing, we introduce the convention
\begin{align}
\tau &:= (-1)^d e^{2\pi i /d},
    &
 D &:= 
    \begin{cases} 
        2d & \text{if } d=2 \\
        d & \text{else.}
    \end{cases}
\end{align}
Note that $\tau$ is always a $D$-th root of unity such that $\tau^2=\omega$.    
We group the $Z$ and $X$ operators and their coordinates to define an arbitrary (generalised) Pauli operator indexed by $a=(a_z,a_x)\in\F_d^{2n}$:
\begin{align}
 \label{eq:weyl-ops}
 w(a) &:= \tau^{-\gamma(a)} Z(a_z) X(a_x), & \gamma(a) &:= a_z\cdot a_x \mod D.
\end{align}

Finally, the \emph{Heisenberg-Weyl} or \emph{generalised Pauli group} is the group generated by Pauli operators and can be written as:
\begin{equation}
 \Pa_n(d) := \langle\{ w(a) \; | \; a\in\F_d^{2n} \}\rangle = \{ \tau^k w(a) \; | \; k\in\Z_D, a\in\F_d^{2n} \}.
\end{equation}
The \emph{Clifford group} is defined as the group of unitary symmetries of the Pauli group:
\begin{equation}
 \label{eq:def-clifford-group}
 \Cl_n(d) := \big\{ U \in U(d^n) \; | \; U\Pa_n(d) U^\dagger = \Pa_n(d) \big\} \, / \, U(1).
\end{equation}
We take the quotient with respect to irrelevant global phases in order to render the Clifford group a finite group. 
If the dimension $d$ is clear from the context, we often omit it to simplify notation.

\begin{figure}[b]
\begin{minipage}[c]{0.5\textwidth}
\centering
\begin{tikzpicture}[z=-5,scale=2]
\coordinate (X1) at (1,0,0);
\coordinate (X2) at (-1,0,0);
\coordinate (Z1) at (0,1,0);
\coordinate (Z2) at (0,-1,0);
\coordinate (Y1) at (0,0,1);
\coordinate (Y2) at (0,0,-1);
\coordinate (H) at (0.71,0,0.71);
\coordinate (T) at (0.5774,0.5774,0.5774);

\draw (X1) -- (Y1);
\draw (X2) -- (Y1);

\begin{scope}[dashed,opacity=0.6]
\draw (Y2) -- (Z1);
\draw (Y2) -- (Z2);
\draw (X2) -- (Y2);
\draw (X1) -- (Y2);
\end{scope}

\draw (X1) -- (Z1);
\draw (X2) -- (Z1);
\draw (Y1) -- (Z1);

\draw [fill=black!15!white,opacity=0.6] (X1) -- (Z2) -- (Y1) -- (X1);
\draw [fill=black!30!white,opacity=0.6] (Y1) -- (X2) -- (Z2) -- (Y1);
\draw [fill=black!15!white,opacity=0.6] (X1) -- (Z1) -- (Y1) -- (X1);
\draw [fill=black!30!white,opacity=0.6] (Y1) -- (X2) -- (Z1) -- (Y1);

\draw (X2) -- (Z2);
\draw (Y1) -- (Z2);
\draw (X1) -- (Z2);
   
\node [above] at (Z1) {\footnotesize $\ket{+Z}$};
\node [below] at (Z2) {\footnotesize $\ket{-Z}$};
\node [right] at (X1) {\footnotesize $\ket{+X}$};
\node [below left] at (Y1) {\footnotesize $\ket{+Y}$};

\end{tikzpicture}
\end{minipage}
\begin{minipage}[c]{0.45\textwidth}
\caption{Bloch representation of the single-qubit \emph{stabiliser polytope}, which is the octahedron spanned by the six $\pm 1$ eigenstates of the Pauli $X$,$Y$, and $Z$ operators. The simple geometry is not representative for the general situation in high dimensions.}
\label{fig:stab_polytope}
\end{minipage}
\end{figure}
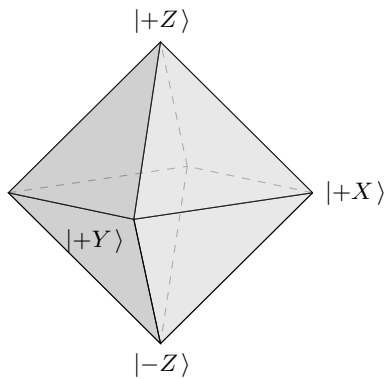

An Abelian subgroup $S\subset\Pa_n(d)$  that does not contain $\omega\one$ is called a \emph{stabiliser group}.
The subspace $C(S)\subset(\C^d)^{\otimes n}$ of common fixed points of $S$ is the \emph{stabiliser code} associated with $S$.
One verifies easily that the orthogonal projection onto $C(S)$ is given by 
\begin{align}\label{eqn:stab projection}
				P_S = |S|^{-1} \sum_{s\in S} s.
\end{align}
By taking traces, it follows that the dimension $\dim C(S)$ equals $d^n/|S| = d^{n-k}$, where $k=\rank(S)$ is the rank of $S$.
Hence, $S$ defines a $[[n,n-k]]$ quantum code and we denote by $\stab(d,n,k)$ the set of these stabiliser codes.
Of particular interest is the case $k=n$, for which $P_S$ is rank 1 and thus defines a pure quantum state, called \emph{stabiliser state}.
The set of pure stabiliser states $\stab(d,n)\equiv\stab(d,n,n)$ spans a convex polytope that is full-dimensional in state space, the \emph{stabiliser polytope} $\SP_n(d):=\conv\stab(d,n)$.
For a single qubit, i.e.~$d=2$ and $n=1$, this is the well-known octahedron spanned by the Pauli $X,Y,Z$ eigenstates, see Fig.~\ref{fig:stab_polytope}.
Elements of $\SP_n(d)$ will be referred to as \emph{mixed stabiliser states}.

%---------------------------------
\subsection{Stabiliser operations}
%---------------------------------
\label{sec:stab-ops}

The Gottesman-Knill theorem states that \emph{stabiliser operations} can be simulated in a time which is polynomial in the system size \cite{gottesman_stabilizer_1997,aaronson_improved_2004}. 
These operations are defined as follows.

\begin{definition}[Stabiliser operation]
 \label{def:so}
 A quantum channel taking $n$ input qudits to $m$ output qudits, each of prime dimension $d$, is a \emph{stabiliser operation}, if it is composed of the following fundamental operations:
 \begin{itemize}
  \item preparation of qudits in stabiliser states,
  \item application of Clifford unitaries,
  \item Pauli measurements, and
  \item discarding of qudits.
\end{itemize}
An arbitrary random function of previous measurement outcomes can be used to decide which fundamental operation to perform in each step.
The set of all stabiliser operations is denoted by $\SO_{n,m}(d)$, with $\SO_n(d) := \SO_{n,n}(d)$.
If the dimension $d$ is clear from the context, we often omit it to simplify notation.
\end{definition}

Typically, one requires that the classical control logic can be implemented in a computationally efficient way
(and the Gottesman-Knill Theorem applies only under this additional assumption).
In the present paper we will drop the efficiency requirement and show that even the resulting larger class of stabilizer operations is smaller than the set of CSP channels.
As we lay out in Remark~\ref{rem:classical-control}, this strengthening of the problem formulation is actually necessary in order to avoid a trivial separation of SO and CSP due to their different computational capabilities.

Because of the possibility to make use of randomness, the set of stabiliser operations $\SO_{n,m}$ is convex.
Its extreme points will turn out to play an important role in our construction.

By definition, stabiliser operations can be seen as an iterative protocol where a quantum computer capable of performing fundamental stabiliser operations interacts with a classical control logic.
Generalising results on the structure of Kraus operators of stabiliser operations obtained in Ref.~\cite{campbell_structure_2009}, we will establish in Thm.~\ref{thm:kraus-decomposition-SO} that any operation in $\SO_{n,m}$ requires at most $n$ interactive rounds.
This stands in contrast to the class LOCC studied in entanglement theory, where no analogous finite bound exists \cite{chitambar_local_2011}.

In our analysis, we will come across the class of stabiliser operations that involve no measurements or classical randomness.
This class coincides with the set of channels whose dilation can be realized with a Clifford unitary:

\begin{definition}
A superoperator $\mathcal{E}:\, L((\C^d)^{\otimes n})\rightarrow L((\C^d)^{\otimes m})$ has a \emph{Clifford dilation} if there exists 
a number $k$, 
a $k$-qudit stabiliser state $|s\rangle$,
and 
a Clifford unitary $U$ on $n+k$ qudits 
such that
\begin{align*}
				\mathcal{E}(\rho) = \tr_{m+1,\dots,n+k} \big[ U (\rho\otimes\ketbran{s}{s}) U^\dagger \big].
\end{align*}
\end{definition}

%-----------------------------------------------------
\subsection{Completely stabiliser-preserving channels}
%-----------------------------------------------------
\label{sec:csp-channels}

From a resource-theoretic perspective, the maximal set of free operations is the set of quantum channels which do not generate resources, \ie~which preserve the set of free states, see \eg~Ref.~\cite{chitambar_quantum_2019}.
If we take the set of free states to be the \emph{stabiliser polytope} $\SP_n(d)$, the resource non-generating (RNG) channels are the \emph{stabiliser-preserving} (SP) \emph{channels}.
For this maximal set of free operations, relatively strong statements can be made from general resource-theoretic arguments.
For instance, it has been recently shown that the resource theory with stabiliser-preserving channels is asymptotically reversible which implies that resource-optimal distillation rates can be achieved with stabiliser-preserving channels \cite{liu_many-body_2020}.

In general, a resource theory with RNG channels has the disadvantage that it is not closed under tensor products since RNG channels may fail to be free when applied to subsystems.
The class of RNG channels for which this is still the case are the \emph{completely} resource non-generating channels \cite{chitambar_quantum_2019}.
For some resource theories, these two classes coincide, but not for the resource theory of magic \cite{seddon_quantifying_2019}.

Following this idea, \textcite{seddon_quantifying_2019,seddon_quantifying_2021} have studied completely stabiliser-preserving (CSP) channels as the free operations in a resource theory of magic state quantum computing.

\begin{definition}
\label{def:csp}
	A superoperator $\mathcal{E}:\, L((\C^d)^{\otimes n})\rightarrow L((\C^d)^{\otimes m})$ is called \emph{completely stabiliser-preserving} (CSP) if and only if
	$\mathcal{E}\otimes\id_k(\SP_{n+k}(d))\subset \SP_{m+k}(d)$ for all $k\in\N$.
	The set of CSP maps is denoted by $\CSP_{n,m}(d)$ and $\CSP_n(d):=\CSP_{n,n}(d)$.
    If the dimension $d$ is clear from the context, we often omit it to simplify notation.
\end{definition}
As it is the case for completely positive maps, one can show that it is indeed enough to check the condition for $k=n$ \cite[Lem.~4.1]{seddon_quantifying_2019}.

It will be helpful to characterise CSP maps via their 
\emph{Choi-Jamiołkowski representation}.
Recall that in this representation, a linear map $\mathcal E: L((\C^d)^{\otimes n}) \to L((\C^d)^{\otimes m})$ is associated with an operator
\begin{equation}
 \mathcal{J}(\mathcal E):= \mathcal{E}\otimes\id_{n}(\ketbrab{\phi^+}{\phi^+}) \in 
  L((\C^d)^{\otimes m}) \otimes L((\C^d)^{\otimes n}),
\end{equation}
where $\ketb{\phi^+} = d^{-n} \sum_{x\in\F_d^n} \ket{xx}$ is the standard maximally entangled state with respect to the computational basis. 
Choi's theorem states that $\mathcal{E}$ is completely positive if and  only if its Choi-Jamiołkowski representation lies in the positive semidefinite cone 
\begin{equation}
  \mathrm{PSD}_{n+m}(d) \subset L((\C^d)^{\otimes m}) \otimes L((\C^d)^{\otimes n}).
\end{equation}
What is more, the map $\mathcal{E}$ is trace-preserving if and only if its Choi-Jamiołkowski representation lies in the affine space
\begin{equation}
  \mathrm{TP}_{n,m}(d)
  = \left\{ \rho \in L((\C^d)^{\otimes m}) \otimes L((\C^d)^{\otimes n}) \,|\, \tr_1 \rho = \one/d^m \right\}.
\end{equation}
In particular, for the set $\mathrm{CPTP}_{n,m}(d)$ of completely positive and trace-preserving maps, we have the characterization
\begin{align}\label{eqn:choi-jam}
  \mathcal{J}( \mathrm{CPTP}_{n,m}(d) ) = \mathrm{PSD}_{n+m}(d) \cap \mathrm{TP}_{n,m}(d).
\end{align}

We now turn to the CSP version of this theory.
It turns out that the CSP property has strong implications:
\begin{lemma}
  Any CSP map is completely positive and trace-preserving.
\end{lemma}

\begin{proof}
  The first claim follows from the Choi-Jamiołkowski Theorem, because $\ket{\phi^+}$ is a stabiliser state.
  As for the second claim: Because the set of stabiliser states (as projections) spans $L((\C^d)^{\otimes n})$, every Hermitian trace-one operator can be written as an affine combination of stabiliser states. 
  By definition, any CSP map maps this to an affine combination of stabiliser states in the output space $L((\C^d)^{\otimes m})$.
  In particular, it is trace-preserving.
\end{proof}

The CSP-analogue of Eq.~\eqref{eqn:choi-jam} was proven in Ref.~\cite{seddon_quantifying_2019}.
\begin{lemma}[Lem.~4.2 in \cite{seddon_quantifying_2019}]
\label{lem:csp-choi}
A linear map $\mathcal E: L((\C^d)^{\otimes n}) \to L((\C^d)^{\otimes m})$ is CSP if and only if its Choi representation lies in the intersection of the stabiliser polytope with the affine space $\mathrm{TP}_{n,m}(d)$:
\begin{equation}
  \mathcal{J}(\mathrm{CSP}_{n,m}(d)) = \mathrm{SP}_{n+m}(d) \cap \mathrm{TP}_{n,m}(d).
\end{equation}
In particular, $\CSP_{n,m}(d)$ is a convex polytope.
\end{lemma}

Additional properties of CSP channels, as well as a collection of examples, are provided in Sec.~\ref{sec:additional}.

For this work, the focus lies on channels which map the input space to itself, i.e.~$n=m$.
In the main part of this paper, we study the relation between completely stabiliser-preserving channels $\CSP_n(d)$ and stabiliser operations $\SO_n(d)$.
In particular, we show that they agree if and only if $n=1$.
The definitions in this section, as well as the general version of our main result, apply both to qubits $d=2$ and to qudits, where $d$ is an odd prime number.

However, we point out that in odd prime dimensions, the set of free states can be enlarged to include all states with a non-negative Wigner function.
This is a convex set $\mathcal{W}_n^+(d)$ given as the intersection of a probability simplex with the cone of positive-semidefinite matrices, and strictly larger than the stabiliser polytope \cite{gross_hudsons_2006,gross_non-negative_2007}.
The resulting resource theory differs quite significantly from the qubit case \cite{veitch_negative_2012,veitch_resource_2014,mari_positive_2012} and naturally leads to a different class of resource-non generating channels, namely those which do not induce Wigner negativity, see \eg~Ref.~\cite{wang_quantifying_2019}.
Thus, the questions we ask are arguably better motivated in the qubit case.

Another difference between the resource theories in even and odd dimensions is given by \textcite{ahmadi_quantification_2018}.
They show that for a single qutrit, there is a stabiliser-preserving channel which can induce negativity in a state's Wigner function, in particular it cannot be a stabiliser operation.
This shows that SP channels are not the correct free operations for a resource theory of magic in odd dimensions.
In contrast, we show in this work that the set of \emph{completely} stabiliser-preserving channels agrees with the set of stabiliser operations for a single qudit, independent of the dimension.
Moreover, arbitrary multi-qudit CSP channels for odd $d$ cannot induce negativity in the Wigner function by the following argument.
Analogous to Lemma \ref{lem:csp-choi}, one can show that the set $\mathrm{CWPP}_{n,m}(d)$ of completely $\mathcal{W}^+_n(d)$-preserving channels corresponds to $\mathcal{W}^+_{n+m}(d) \cap \mathrm{TP}_{n,m}(d)$.
This follows from the Choi-Jamiołkowski inversion formula and the fact that $\ketbra{\phi^+}{\phi^+}\in\SP_{n+m}(d)\subset \mathcal{W}^+_{n+m}(d)$.
Therefore, $\CSP_{n,m}(d)$ is contained in $\mathrm{CWPP}_{n,m}(d)$ and, in particular, cannot induce negativity in the Wigner function.
This establishes the chain of inclusions $\SO_{n,m}(d) \subset \CSP_{n,m}(d) \subset \mathrm{CWPP}_{n,m}(d)$ for odd prime $d$, where our main result \ref{thm:main} implies that the first inclusion is proper for $n,m>1$.
While one cannot readily dismiss the possibility that the last inclusion is an equality, we conjecture that it is indeed proper, too.

\begin{remark}
\label{rem:classical-control}
The definition \ref{def:csp} of CSP allows for quantum channels of the form \cite{campbell}
\begin{equation}
 \mathcal{E}(\ketbra{x}{y}) := \delta_{x,y} \ketbra{O(x)}{O(x)} 
\end{equation}
where $O$ can be an arbitrary Boolean function.
The definition does not preclude one to consider families $\mathcal{E}_n$ of channels that are associated with Boolean functions $O_n$ that are not Turing computable (e.g.\ functions that decide the halting problem).
The discussion shows that it is meaningless to compare stabilizer operations \emph{with} computational efficiency requirements to CSP channels defined \emph{without} such constraints.
To avoid a trivial separation of the classes, we show here that even stabilizer operations where the classical control logic can consist of arbitrary random functions of previous measurement results cannot implement all CSP channels.
\end{remark}

%%% =============================================
\section{The CSP class is strictly larger than the class of stabiliser operations}
\label{sec:min-result}
%%% =============================================

In this section, we prove a minimal version of the main result.
The general version, treating the multi-qudit case, is stated and proven in Sec.~\ref{sec:gen-results}.

\begin{theorem}
\label{thm:main-minimal}
	  For two qubits, the set $\CSP_2(2)$ is strictly larger than $\SO_2(2)$.
\end{theorem}

Concretely, we will establish that the following two-qubit channel is completely stabiliser-preserving, but not a stabiliser operation:
\begin{align}\label{eq:lambda-2-comp-basis}
	\Lambda(\rho) := 
	\rho_{00,00} \ketbra{++}{++} +  
	\sum_{x\in\{01,10,11\}} \rho_{x,x} \ketbra{x}{x} + 
	\frac 1 2 \sum_{\substack{x, y \in \{01,10,11\}\\x\neq y}} \rho_{x,y} \ketbra{x}{y},
\end{align}
where $\ket+=\frac1{\sqrt 2}(\ket 0 + \ket 1)$ and $\rho_{x,y} = \langle x|\rho|y \rangle$.

The intuition behind the counter-example is as follows:
First, consider a projective measurement that distinguishes between $\ket{00}$ and its orthocomplement.
It is plausible that one cannot implement such a measurement using stabiliser operations -- if for no other reason than that Pauli measurements lead to Kraus operators whose rank is a power of two.
The channel $\Lambda$ may be realized by such a measurement, followed by the application of Hadamard gates on all qubits when the outcome $\ket{00}$ is obtained, or a partial dephasing operation in the alternate case.
It turns out that the second step makes $\Lambda$ CSP, while the no-go argument concerning the measurement remains valid.

Appendix~\ref{sec:on-lambda} describes some properties of $\Lambda$ that are not directly required for the proof below.

In the proof, we will use the fact that the channel (\ref{eq:lambda-2-comp-basis}) is an extreme point in the convex set $\CSP_2 \equiv \CSP_2(2)$.
To show this, it turns out to be sufficient to restrict attention to the intersection of $\CSP_2$ with a fairly low-dimensional affine space -- a step that greatly simplifies the description of the convex geometry.

Concretely, we define the convex set of \emph{almost-diagonal channels} $\AD_2$ as the set of two-qubit quantum channels $\mathcal{E}$ that act on the pure states of the computational basis in the following way:
\begin{align}
	\mathcal{E}(|00\rangle\langle 00|) =  \ket{++}\bra{++}, \qquad
	\mathcal{E}(|x\rangle\langle x|) =  |x\rangle\langle x| \quad x\in\{01,10,11\}.
\label{eq:def-AD} 
\end{align}
By comparison with Eq.~\eqref{eq:lambda-2-comp-basis} it is immediate that $\Lambda$  lies in $\CSP_2\cap \AD_2$.
This intersection is isomorphic, as a convex set, to a subpolytope of the two qubit stabiliser polytope.

\begin{definition}
Let 
$\mathrm{P}_{2}$ 
be the polytope of complex $4\times 4$ matrices $\sigma$ that 
(1) are a convex combination of two-qubit stabiliser states, and 
(2), when expressed in the $\{\ket{00},\ket{01},\ket{10}, \ket{11}\}$-basis, are of the form
\begin{align*}
	\sigma = 
	\frac13
	\begin{pmatrix}
		0 & 0 & 0 & 0 \\
		0 & 1 & * & * \\
		0 & * & 1 & * \\
		0 & * & * & 1 
	\end{pmatrix}
\end{align*}
with $*$'s denoting arbitrary complex values.
\end{definition}

\begin{lemma}\label{lem:diag-subpolytope}
A map $\mathcal{E}$ lies in $\CSP_2 \cap \AD_2$ if and only if there exists a $\sigma\in\mathrm{P}_2$ such that
\begin{align}
				\mathcal{E}(\rho) = 3\,\sigma \circ \rho + \langle 00|\rho|00\rangle \, \ket{++}\bra{++},
\end{align}
where $\circ$ is the \emph{Hadamard} (or element-wise) product.
In particular, the polytopes $\CSP_2\cap \AD_2$ and $\mathrm{P}_2$ are isomorphic.
\end{lemma}

\begin{proof}
\textbf{``Only if'':}
Assume that $\mathcal{E}$ is CSP, i.e.\ its Choi state is expressible as
\begin{align}\label{eqn:lambda decomp}
    \mathcal{J}(\mathcal{E}) = \sum_{s\in\stab(4)} p_s |s\rangle\langle s|.
\end{align}
The Choi state has the property that
\begin{align}\label{eqn:choi property}
    \mathcal{E}(|x\rangle\langle y|)
    =
    4 
			(\one\otimes\bra x)\,
    \mathcal{J}(\mathcal{E}) 
    \,(\one\otimes \ket y)
    = 
    4 
			\sum_s p_s
    \,
    (\one\otimes\bra x) |s\rangle
    \,
    \langle s|(\one\otimes \ket y)
    \quad
    \forall x,y \in \FF_2^{2}.
\end{align}
Evaluating Eq.~(\ref{eqn:choi property}) on the diagonal and using Eq.~(\ref{eq:def-AD})  implies that, for all $s$ with $p_s\neq 0$,
\begin{align}
	(\one\otimes\bra{00}) |s\rangle  &\propto \ket{++}, \label{eqn:0 to +} \\
	(\one\otimes\bra{x}) |s\rangle  &\propto |x \rangle \quad \forall x \neq 00, \label{eqn:x to x} 
\end{align}
where $\propto$ denotes equality up to a proportionality constant including $0$.

There must be at least one $|s\rangle$ with $(\one\otimes\bra{00}) |s\rangle \neq 0$.
We claim that this implies $|s\rangle=|{++}\rangle|{00}\rangle$ and $p_s = \frac14$.
Indeed, assume for the sake of reaching a contradiction that $\ket s$ has Schmidt rank larger than one.
Then for at least one $x\in\FF_2^2$, the contraction $(\one\otimes\bra{x}) |s\rangle $ is not proportional to $\ket{++}$.
By a well-known property of stabiliser states (c.f.~Prop.~\ref{prop:well-known stab basis}), $(\one\otimes\bra{x}) |s\rangle$ is then orthogonal to $|++\rangle$, which contradicts (\ref{eqn:x to x}).
Thus $|s\rangle$ is a product state. 
The claimed form follows from (\ref{eqn:0 to +}), and the value of $p_s$ from (\ref{eqn:choi property}). 

We now treat the terms $\ket{s}$ different from $\ket{++}\ket{00}$.
Equations~(\ref{eqn:0 to +}, \ref{eqn:x to x}) and Proposition~\ref{prop:well-known stab basis} imply that these stabiliser states are ``diagonal in the computational basis'' in the sense that
\begin{align*}
|s\rangle = 
\sum_{x\in\FF_2^2} \tilde s(x) |x \rangle |x\rangle
\qquad\text{for some $\tilde s: \FF_2^2\to\C$ with $\tilde s(00)=0$}.
\end{align*}
Define the $n$-qudit state $|\tilde s\rangle=\sum_x \tilde s(x)\ket x$.
Then $\ket{\tilde s}$ is orthogonal to $\ket{00}$.
It is also a normalised stabiliser state, because it arises from the action of a Clifford unitary on $\ket{s}$:
\begin{align*}
				\ket{\tilde s}\otimes\ket{00} = CX_{1,3}\,CX_{2,4} \ket{s},
\end{align*}
where $CX_{i,j}$ is the controlled-NOT gate with the $i$-th qubit controlling the $j$-th one.
Setting
\begin{align*}
    \sigma=
			\frac43
			\, \sum_{s\neq |{++}\rangle|00\rangle} p_s \ketbra{\tilde s}{\tilde s} \in \mathrm{P}_{2}, 
\end{align*}
we get that for all $(x,y)\neq(00,00)$,
\begin{align*}
    \mathcal{E}(\ketbra{x}{y})
			&=
			4
    \sum_s
    p_s
    \,
    (\one\otimes\bra x) \ketbra{s}{s} (\one\otimes \ket y) \\
			&=
			4
			\sum_{s\neq \ket{++}\ket{00}}
    p_s
    \,
    \tilde s(x)
    \overline{\tilde s(y)}
    \ketbra{x}{y} \\
			&=
			3
			\,
    \sigma \circ
    \ketbra{x}{y}.
\end{align*}

\textbf{``If'':} The construction above can be reversed straight-forwardly.
\end{proof}

Under the correspondence given in Lemma~\ref{lem:diag-subpolytope}, the channel $\Lambda$ defined in Eq.~(\ref{eq:lambda-2-comp-basis}) corresponds to the matrix
\begin{align}\label{eqn:lambda-def}
		\lambda=
		\frac16
		\begin{pmatrix}
		0 & 0 & 0 & 0 \\
		0 & 2 & 1 & 1 \\
		0 & 1 & 2 & 1 \\
		0 & 1 & 1 & 2
		\end{pmatrix},
\end{align}
Using the relative simplicity of the polytope $\mathrm{P}_{2}$, we can now show that $\Lambda$ is an extremal CSP channel.

\begin{lemma}
\label{lem:lambda-extremality-AD-qubit}
The matrix $\lambda$ in Eq.~(\ref{eqn:lambda-def}) is a vertex of $\mathrm{P}_2$.
What is more, $\Lambda$ is a vertex of $\CSP_2$.  
\end{lemma}

\begin{proof}
We will establish the first claim by showing that $\lambda$ is the unique maximiser in $\mathrm{P}_2$ of the linear functional
\begin{align*}
	L: \mathrm{P}_2 \to \R,	\qquad
    \sigma
    \mapsto
    \langle +| \sigma | + \rangle
    =
    \sum_s p_s |\langle +| s\rangle|^2. 
\end{align*}
There are 15 stabiliser states $\ket{s}$ orthogonal to $\ket{00}$, given by
\begin{align*}
    \ket{01}, \qquad 2^{-1/2}\big(\ket{01} + \omega \ket{10}\big), \quad \omega\in\{1, -1, i, -i \},
\end{align*}
and their images under permutations of $\{ \ket{01}, \ket{10}, \ket{11} \}$.
	Among those, the inner product $|\langle +|s\rangle|^2$ attains its maximum (of $1/2$)
	exactly for the cases
$2^{-1/2}\big(\ket{01} +\ket{10}\big),
2^{-1/2}\big(\ket{01} + \ket{11}\big),
2^{-1/2}\big(\ket{10} + \ket{11}\big)$.
Among the linear combinations of their projection operators, a uniform mixture is the unique solution to the three constraints $\sigma_{x,x}=1/3$.
This solution is equal to $\lambda$.

To prove the second claim, assume that $\AD_2\cap\CSP_2\ni\mathcal{E} = p \mathcal{E}_1 + (1-p)\mathcal{E}_2$ for some CSP maps $\mathcal{E}_1$ and $\mathcal{E}_2$, and $p\in [0,1]$.
The extremality of the pure states on the right hand sides of Eq.~\eqref{eq:def-AD} forces $\mathcal{E}_1$ and $\mathcal{E}_2$ to fulfil the same constraints, i.e.~$\mathcal{E}_1,\mathcal{E}_2\in\AD_2\cap\CSP_2$.
Hence, a channel $\mathcal{E}\in\AD_2\cap\CSP_2$ is extremal in $\CSP_2$ if and only if it is extremal in the subpolytope $\AD_2\cap\CSP_2$.
\end{proof}

If $\Lambda$ was a stabiliser operation, Lem.~\ref{lem:lambda-extremality-AD-qubit} would imply that it is extremal in the convex set $\SO_2$.
This is because extremality of a point in a convex set implies extremality in every convex subset containing the point.
Our strategy now is to identify a property shared by all extremal stabiliser operations, and then to show that $\Lambda$ fails to posses it.

\begin{theorem}[Pauli invariance of extremal stabiliser operations]
 \label{thm:stab-op-invariance}
 Let $\mathcal O\in\SO_{2}$ be an extremal stabiliser operation that does not have a Clifford dilation.
 Then the kernel of $\mathcal O$ contains a Pauli operator.
\end{theorem}

The proof will make use of the following lemma.	
It says that the operation ``preparing an ancilla stabiliser state and performing a Pauli measurement jointly on an input and the ancilla'' can be replaced by a random Clifford channel, if the stabiliser state is not an eigenstate of the Pauli operator.
(One could also approach the statement through the theory of quantum error correction.
In this language, the measured Pauli is a correctable error for the stabiliser code $(\CC^2)^{\otimes n}\otimes |s\rangle$, and the Clifford unitaries that appear are the ones correcting the projections onto the eigenspaces of the Pauli operator.)

\begin{lemma}
\label{lem:inv-case-two} 
	Let $w(a)\otimes w(b)$ be an $(n+k)$-qubit Pauli operator.
Denote the projectors onto the two eigenspaces of $w(a)\otimes w(b)$ by $P_\pm$.
	Let $\ket s$ be a $k$-qubit stabiliser state that is \emph{not} an eigenstate of $w(b)$.
	Then there are two $(n+k)$-qubit Clifford unitaries $U_\pm$ such that, for all $n$-qubit states $|\psi\rangle$, we have 
	$P_\pm \ket\psi \otimes \ketn{s} = \frac1{\sqrt{2}} U_\pm \ket\psi \otimes \ketn{s}$.
\end{lemma}

\begin{proof}
	There is a $k$-qubit Clifford unitary $V$ which maps 
	$\ketn{s}\mapsto\ketn{0^k}$ and $(w(b)\ketn{s})\mapsto \ketn{1}\ketn{0^{k-1}}$.
  There is also an $n$-qubit Clifford $U$ such that $U w(a) U^\dagger=Z_1$.
	It thus suffices to show the claim for the special case $w(a)=Z_1$ and $w(b)\ketn{0^k} = \ketn{1}\ketn{0^{k-1}}$.
	In terms of a controlled $Z$-gate $CZ_{(n+1),1}$ (first ancilla qubit controlling the first input qubit):
  \begin{equation}
	P_\pm (\ket{\psi}\otimes\ketn{0^k}) = \frac12\Big[\ket\psi\otimes \ket 0 \pm (Z_1\ket\psi)\otimes \ket{1}\Big]\otimes \ketn{0^{k-1}}
	=
	\frac1{\sqrt 2} \Big[CZ_{(n+1),1} \ket\psi\ket\pm\Big] \otimes \ketn{0^{k-1}}.
  \end{equation}
	Hence, we can choose $U_+ = CZ_{(n+1),1}H_{n+1}$ and $U_- = CZ_{(n+1),1}H_{n+1}X_{n+1}$ where $H_{n+1}$ and $X_{n+1}$ are the Hadamard and $X$ gate acting on the first ancilla qubit, respectively.
\end{proof}

\begin{proof}[Proof of Theorem~\ref{thm:stab-op-invariance}]
Consider an implementation of $\mathcal{O}$ using elementary Clifford operations.
By extremality, we may assume that no classical randomness is used.
Thus the implementation must contain at least one Pauli measurement, for else $\mathcal{O}$ would have a Clifford dilation.
Propagating the first Pauli measurement past preceding Clifford unitaries if necessary, there is no loss of generality in assuming that the implementation starts by preparing $k$ ancilla qubits in a stabiliser state $|s\rangle$ and then immediately measures an $(n+k)$-qubit Pauli operator $w(a)\otimes w(b)$ with $a\in\F_2^{2n}$ and $b\in\F_2^{2k}$. 

We will show now that one may in fact assume that $a\neq 0$ and $b=0$, i.e.\ that the implementation starts by measuring a non-trivial Pauli without involving the ancillas.

Indeed, if $a=0$, the measurement only acts on the ancilla systems.
We can thus write $\mathcal O = p_1 \mathcal O_1 + p_{-1} \mathcal O_{-1}$, where $\mathcal O_\pm$ are the operations conditioned on the outcome, and the probabilities $p_\pm$ do not depend on the input state.
Extremality implies that either $\mathcal O_1 = \mathcal O_{-1}$ or only one of the $p_\pm$ differs from $0$.
Hence one can eliminate the measurement from the implementation and restart the proof.
Iterating this argument if necessary, we will eventually obtain a Pauli measurement with $a\neq 0$, as $\mathcal{O}$ does not have a Clifford dilation.

Next assume that $b\neq 0$.
First consider the case where $\ket s$ is not an eigenstate of $w(b)$.
By Lemma~\ref{lem:inv-case-two}, the measurement can be replaced by a process that applies one of two Clifford unitaries, each with probability $1/2$.
Arguing as above, this process either contradicts extremality or can be eliminated.
Thus we may assume that $\ket s$ is an eigenstate of $w(b)$. 
In this case, the ancilla system affects the measurement process only by changing the labels of the measurement results (specifically by multiplying them with the eigenvalue).
Absorbing this deterministic relabelling into any classical control, we may set $b=0$.

Let $P_\pm =\frac12(\one \pm w(a))$ be the projections onto the eigenspaces of $w(a)$.
Choose any two-qubit Pauli operator $w(u)$ that anti-commutes with $w(a)$.
Using the above expression for $P_\pm$, one finds that
\begin{equation*}
			P_+ w(u) P_+ + P_{-} w(u) P_{-} = 0
\end{equation*}
and thus $w(u)\in\ker\mathcal{O}$.
\end{proof}

The following lemma is thus sufficient to establish Theorem~\ref{thm:main-minimal}.
 
\begin{lemma}
\label{lem:lambda-invariance}
	Let $\Lambda$ be as in Eq.~\eqref{eq:lambda-2-comp-basis}.
	 Then $\Lambda$ has no Clifford dilation, and $\ker\Lambda$ does not contain a Pauli operator.
\end{lemma}

\begin{proof}
	Assume, for the sake of reaching a contradiction, that $\Lambda$ does have a Clifford dilation.  
	Then $\Lambda^\dagger$ maps Pauli operators to Pauli operators, up to a phase. 
	From Eq.~\eqref{eq:lambda-2-comp-basis}:
    \begin{align}
     0 &= \sandwich{++}{Z_1}{++} = \tr\left( \ketbra{00}{00} \Lambda^\dagger(Z_1) \right), \label{eq:lambda-dilation-1}\\
	 (-1)^x &= \sandwich{xy}{Z_1}{xy} = \tr\left( \ketbra{xy}{xy} \Lambda^\dagger(Z_1) \right), \quad 
	 |xy\rangle \in \{|01\rangle, |10\rangle, |11\rangle\}.
	 %\forall (x,y)\in\F_2^2\setminus 0. 
	 \label{eq:lambda-dilation-2}
    \end{align}
	Eq.~\eqref{eq:lambda-dilation-1} implies that $\Lambda^\dagger(Z_1)$ is proportional to $X$ or $Y$ on at least one of the factors.
	This, however, is incompatible with Eq.~\eqref{eq:lambda-dilation-2}, which is the sought-for contradiction.

	One reads off Eq.~\eqref{eq:lambda-2-comp-basis} that $\Lambda(\ketbra{x}{y})=0$ if and only if $x = 0$ and $y\neq 0$ or $x\neq 0$ and $y=0$.
	That means that the kernel of $\Lambda$ consists of the operators of the form
	 \begin{equation}
	    \begin{pmatrix}
	        0 & * & * & * \\
	        * & 0 & 0 & 0 \\
	        * & 0 & 0 & 0 \\
	        * & 0 & 0 & 0 
	    \end{pmatrix},
	 \end{equation}
	 each of which has rank at most $2$.
	 In particular, this rules out Pauli operators.
\end{proof}

%%% =============================================
\section{General formulation}
%%% =============================================
\label{sec:gen-results}

In this section, we generalise Theorem \ref{thm:main-minimal} to our main result: $\CSP_n(d)$ strictly contains $\SO_n(d)$ for any (prime) dimension $d$ and system size $n\geq 2$.

\begin{theorem}[$\SO_n \subsetneq \CSP_n$]
\label{thm:main}
 For any prime dimension $d$, we have $\CSP_n(d) = \SO_n(d)$ if and only if $n=1$.
 In particular, the set of CSP maps is strictly larger than the set of stabiliser operations for $n\geq 2$.
\end{theorem}

The proof of Theorem \ref{thm:main} is accomplished in two parts.
The equality in the case $n=1$ is proven independently in Sec.~\ref{sec:csp-eq-so}.
For the case $n\geq 2$, we start with an identical approach as in Sec.~\ref{sec:min-result} and concentrate on the intersection of $\CSP_n(d)$ with almost-diagonal (AD) channels.
Although the proof strategy of Sec.~\ref{sec:min-result} based on Pauli invariances, i.e.~Theorem \ref{thm:stab-op-invariance}, also works for arbitrary $d$ and $n\geq 2$, we follow a more direct route in this section.
As we show, the restriction to AD channels directly simplifies the description of both general CSP channels and stabiliser operations considerably.
Using this result, it is then straightforward to define a linear functional $L$ which separates the almost-diagonal CSP channels from stabiliser operations.
As we show, this linear functional is again maximal on a generalisation of the $\Lambda$ channel to be defined later, cp.~Eq.~\eqref{eq:lambda-2-comp-basis}.

To arrive at the mentioned simplification for stabiliser operations, we first derive a suitable ``normal form'' in Sec.~\ref{sec:so-normal-form}.

%--------------------------------------------------
\subsection{Normal form for stabiliser operations}
\label{sec:so-normal-form}
%--------------------------------------------------

In this section, we show that any stabiliser operation is a convex combination of circuits performing a projective stabiliser measurement on the input followed by a global, ancilla-assisted Clifford unitary conditioned on the measurement outcome.

\begin{theorem}[Kraus decomposition of SO]
\label{thm:kraus-decomposition-SO}
 Consider the family of stabiliser operations in $\SO_{n,m}(d)$ of the following type:
 \begin{equation}
 \label{eq:SO-syndrome-measurement}
  \mathcal{E}(\rho) = \tr_{m+1,\dots,n+r} \sum_{i} U_i \left( P_i \rho P_i \otimes\ketbran{0^r}{0^r} \right) U_i^\dagger,
 \end{equation}
 where $\{ P_i \}$ is a projective measurement given by mutually orthogonal stabiliser code projectors and the $ U_i $'s are Clifford unitaries acting on $ n+r $ qudits.
 Then, the following holds:
 \begin{enumerate}[label=(\roman*)]
  \item Any $\mathcal{O}\in\SO_{n,m}(d)$ is a convex combination of SO of the above type \eqref{eq:SO-syndrome-measurement}. 
  \item In particular, any stabiliser operation can be realised in at most $n$ rounds.
 \end{enumerate}	 
\end{theorem}

\begin{remark}
 A projective measurement composed of mutually orthogonal stabiliser code projectors is not necessarily associated to a single set of mutually commuting Pauli operators (i.e.~a syndrome measurement).
 An example for this is the measurement of the basis $\{ \ket{00}, \ket{01}, \ket{1+}, \ket{1-} \}$.
\end{remark}

The proof of Theorem \ref{thm:kraus-decomposition-SO} is similar to related results in Ref.~\cite{campbell_structure_2009} and Ref.~\cite[Thm.~5.3]{beverland_lower_2020}.
However, the latter works focus on the form of \emph{post-selected} stabiliser operations, i.e.~on the form of a single Kraus operator in Eq.~\eqref{eq:SO-syndrome-measurement}.
Moreover, Ref.~\cite{beverland_lower_2020} only considers the form of post-selected stabiliser operations which map a fixed input to a fixed output state.
Here, we show that a careful argumentation allows us to manipulate all Kraus operators simultaneously to arrive at a similar result for the entire quantum channel.
The

To prove Theorem~\ref{thm:kraus-decomposition-SO}, we use Lemmata~\ref{lem:non-commuting-codes} and~\ref{lem:inv-case-general} to eliminate non-commuting Pauli measurements and Pauli measurements on ancilla qudits.
In this way, an arbitrary stabiliser operation can be iteratively decomposed into a convex combination of stabiliser operations of the form \eqref{eq:SO-syndrome-measurement}.

Lemma \ref{lem:non-commuting-codes} generalises \cite[{Sec.~6}]{campbell_structure_2009} and \cite[{Prop.~A8}]{beverland_lower_2020} to arbitrary prime dimension $d$.

\begin{lemma}
\label{lem:non-commuting-codes}
 Suppose $P_1$ and $P_2$ are non-commuting $[[n,n-1]]$ stabiliser code projectors.
 Then, $P_1P_2 = d^{-1/2} VP_2$ for a suitable Clifford unitary $V$.
\end{lemma}

\begin{proof}
 Pairs of non-commuting $[[n,n-1]]$ stabiliser codes form a single orbit under the Clifford group.
 To see this, let $w_1$ and $w_2$ be Pauli operators that generate such a pair $P_1$ and $P_2$ and let $\tilde w_1$, $\tilde w_2$ generate another pair $\tilde P_1$ and $\tilde P_2$.
 By redefining the generators with a suitable power of $\omega$, we can assume that $w_1w_2 = \omega w_2w_1 $ and $\tilde w_1 \tilde w_2 = \omega \tilde w_2 \tilde w_1$.
 Then there is a Clifford unitary $U$ mapping $w_1$ to $\tilde w_1$ and $w_2$ to $\tilde w_2$ and thus $P_1$ to $\tilde P_1$ and $P_2$ to $\tilde P_2$ as claimed.
 Thus, we may assume that $P_1=\ketbra{+}{+}\otimes\one_{n-1}$ and $P_2=\ketbra{0}{0}\otimes\one_{n-1}$ and clearly $P_1 P_2 = \ketbra{+}{+} \ketbra{0}{0} \otimes \one_{n-1} = \frac{1}{\sqrt{d}} H P_2$ where $H$ is the Hadamard gate.
\end{proof}

The following lemma is a generalisation of Lemma \ref{lem:inv-case-two} to any prime dimension $ d $.
\begin{lemma}
\label{lem:inv-case-general}
 Let $w(a)\otimes w(b)$ be a $(n+k)$-qudit Pauli operator and let $\ketn{s}$ be a $k$-qudit stabiliser state which is not an eigenstate of $w(b)$.
 For any $x\in\F_d$, denote the projector onto the eigenspace of $w(a)\otimes w(b)$ with eigenvalue $\omega^x$ by $P_x$.
 Then, there are Clifford unitaries $U_x$ such that $P_x \ket\psi \otimes \ketn{s} = d^{-1/2} U_x \ket\psi \otimes \ketn{s} $ for all $\psi\in(\C^d)^{\otimes n}$ and $x\in\F_d$.
\end{lemma}

\begin{proof}
  Since $\ket{s}$ is not an eigenstate of $w(b)$, the stabiliser states $w(xb)\ket{s}$ for $x\in\F_d$ are part of the same stabiliser basis.
  In particular, there is a Clifford unitary $V$ such that $Vw(xb)\ket{s} = \ket{x}\ketn{0^{k-1}}$
  Moreover, there is a Clifford unitary $U$ such that $U w(a) U^\dagger=Z_1$. 
  Thus, up to acting with $U$ on the input register, and with $V$ on the ancilla register, we may assume that $w(a)=Z_1$ and $w(xb)\ketn{0^k} = \ketn{x}\ketn{0^{k-1}}$.
  In terms of a controlled $Z$-gate $CZ_{n+1,1}$ (first ancilla qudit controlling the first input qudit),
  the action of the projections onto the eigenspaces of $w(a)\otimes w(b)$ is then given by
  \begin{equation}
	P_x \ket\psi\otimes\ketn{0^k} = \frac1d \Big[\sum_{y\in\F_d} \omega^{xy} \big( Z_1^x \ket\psi \big) \otimes \ket x  \Big]\otimes \ketn{0^{k-1}}
	=
	\frac{1}{\sqrt d} \Big[CZ_{n+1,1} \big( \ket\psi \otimes H \ket{x} \big) \Big] \otimes \ketn{0^{k-1}}.
  \end{equation}
  Thus, the claim holds for the Clifford unitary $U_x := CZ_{n+1,1} H_{n+1} X_{n+1}(x)$.
\end{proof}

\begin{proof}[Proof of Theorem \ref{thm:kraus-decomposition-SO}]
Suppose $\mathcal{O}$ is a stabiliser operation which does not explicitly use classical randomness and involves $l$ Pauli measurements with outcomes labelled by the ditstring $x=(x_1,\dots,x_l)\in\F_d^l$.    
Let us introduce the shorthand notation $x_{[k]} := (x_1,\dots,x_k)$.
Since the partial trace is linear and the size of the ancilla system stays fixed throughout the proof, we can ignore the possibility of tracing out qudits.
Hence, $\mathcal{O}$ can be taken as follows:
\begin{align}
\label{eq:prop1-kraus-operators}
 \mathcal{O}(\rho) &= \sum_{x\in\F_d^l} K(x)\rho\otimes\ketbran{0^r}{0^r} K(x)^\dagger, &
 K(x) &= U(x) P(x_l\,|\,x_{[l-1]}) P(x_{l-1}\,|\,x_{[l-2]}) \cdots P(x_1).
\end{align}
Without loss of generality, the Kraus operators $K(x)$ are given by consecutive projectors $P$ associated to outcomes of Pauli measurements, and a global Clifford unitary $U$ at the end. 
All operations may be conditioned on previous measurement outcomes.
The projectors fulfil the POVM condition $\sum_{x_k} P(x_{k}\,|\,x_{[k-1]}) = \one$ for all $k$ and previous outcomes $x_{[k-1]}\in\F_d^{k-1}$. 
We can visualise the SO as a regular tree with root given by the initial measurement and branches corresponding to sequences of measurement outcomes.
The vertices of the tree are labelled by Pauli measurements (see Fig.~\ref{fig:SOtree}).

\begin{figure}
	\centering 
	\begin{tikzpicture}[ scale = 0.6,level/.style={sibling distance = 6cm/#1,
			level distance = 1.5cm}] 
		\node [align=center]at (3,0){Initial Measurement}
		child { node  {$ P(0) $} 
			child {node [color = red] {$ P(0|0) $}
				child {node {$ P(0|00) $}}
				child {node {$ P(1|00) $}}
			}	
			child {node  [color = red]{$ P(1|0) $}}
		}
		child { node {$ P(1) $}
			child {node {$ P(0|1) $}}
			child {node {$ P(1|1) $}}
		};
		\node at (-8,-8)  {\Large =};
		\node at (-7,-8) {\huge $ \frac{1}{2} $};
		\node[align=center] at (-2,-6) {Initial Measurement}
		child { node   {$ P(0) $} 
			child {node [color = red] {$ U(00) $}
				child {node {$ P(0|00) $}}
				child {node {$ P(1|00) $}} 
			}	
		}
		child { node {$ P(1) $}
			child {node {$ P(0|1) $}}
			child {node {$ P(1|1) $}}
		};
		\node at (4,-8)  {\Large +};
		\node at (5,-8)   {\huge $\frac{1}{2}$};
		
		\node[align=center] at (9.5,-6){Initial Measurement}
		child { node  {$ P(0) $} 
			child {node [color = red] {$ U(10) $}
				%			child {node {$ P(0|00) $}}
				%			child {node {$ P(1|00) $}} 
			}	
		}
		child { node {$ P(1) $}
			child {node {$ P(0|1) $}}
			child {node {$ P(1|1) $}}
		};
	\end{tikzpicture}
	\caption{Illustration of a tree model of a qubit stabiliser operation. The nodes correspond to outcomes of measurements which in turn depend on the previous outcomes $ x_{[k]} \in \FF_2^k $. We omit all nodes given by trivial measurements.
	If a Pauli measurement acts non-trivially on the ancilla state (in the given case $ P(0|0) $ and $ P(1|0) $ in red), the SO coincides with a uniform convex combination of two SOs, where the measurements $ P(0|0) $ and $ P(1|0) $ are replaced by Clifford unitaries $ U(00) $ and $ U(10) $.}
	\label{fig:SOtree}
\end{figure}
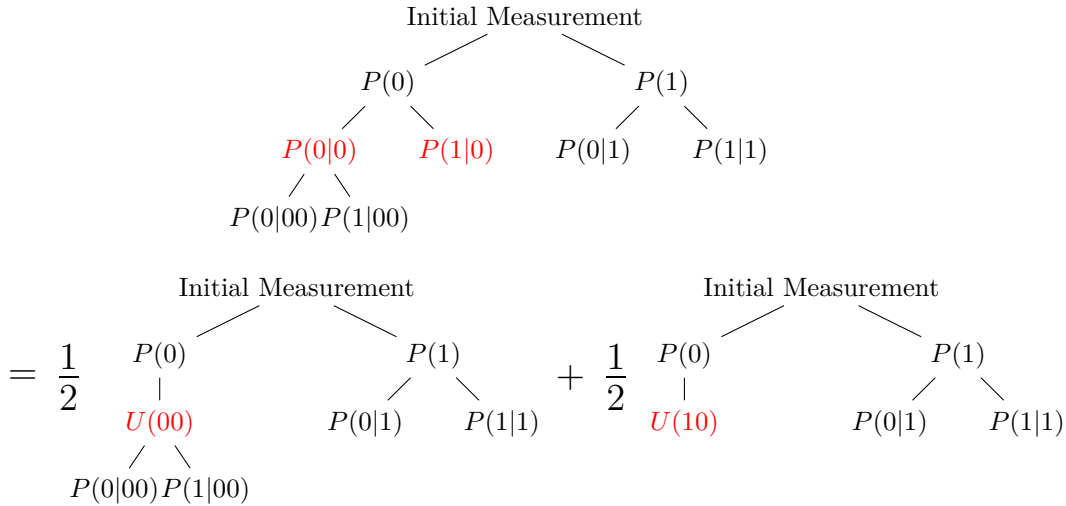

We first argue that we can write $\mathcal{O}$ as a convex combination of stabiliser operations which do not measure ancilla qudits.
To this end, we use Lemma \ref{lem:inv-case-general} to replace any Pauli measurement involving the ancilla system by a convex combination of suitable Clifford unitaries acting on input and ancilla system.
We prove this via induction over the depth $k$ of the tree, starting from the root $k=1$ and progressing to the leaves $k=l$.
Assume that up to depth $k-1$, all measurements are acting trivially on the ancilla system.
This implies that in every branch, the ancilla system is still in the initial state $\ketn{0^r}$.
Let $X_{k-1}$ be the set of previous outcomes $x_{[k-1]}$, such that the $k$-th measurement conditioned on $x_{[k-1]}\in X_{k-1}$ acts non-trivially on ancilla qudits.
By Lemma \ref{lem:inv-case-two}, we can then write $P(x_k|x_{[k-1]})\ket{\psi}\otimes\ketn{0^r} = d^{-1/2} U(x_{[k]})\ket{\psi}\otimes\ketn{0^r}$ for all input states $\ket\psi$ and suitable Clifford unitaries $U(x_{[k]})$.
For branches starting with $x_{[k-1]}$, we can thus treat $x_k$ as the outcome of a classical, uniformly distributed random variable $Y$.
By conditioning on the outcome $y\in\F_d$ of this random variable, we get new stabiliser operations $\mathcal{O}(y)$ given by the Kraus operators
\begin{align}
  K'(x_l,\dots,x_{k+1},y,x_{[k-1]}) &:= d^{1/2} K(x_l,\dots,x_{k+1},y,x_{[k-1]}),
   &x_{[k-1]} \in X_{k-1}, \\
  K'(x_l,\dots,x_{k},x_{[k-1]}) &:= K(x_l,\dots,x_{k},x_{[k-1]}), & x_{[k-1]} \notin X_{k-1}.
 \end{align}
$\mathcal{O}(y)$ performs the same operation as $\mathcal{O}$ if the first $k-1$ outcomes are not in $X_{k-1}$ and otherwise applies the Clifford unitary $U(y,x_{[k-1]})$ and follows the branch determined by $(y,x_{[k-1]})$, see Fig.~\ref{fig:SOtree}.
In particular, it is indeed a stabiliser operation and $\mathcal{O}= d^{-1} \sum_y \mathcal{O}(y)$.
Moreover, all measurements in $\mathcal{O}(y)$ act trivially on the ancilla system up to depth $k$.
We proceed with the induction for $\mathcal{O}(y)$.
This shows that $\mathcal{O}$ is a convex combination of stabiliser operations of the form \eqref{eq:prop1-kraus-operators}, where the measurements do not act on the $k$ ancilla qudits.

Next, let us assume that the measurements in $\mathcal{O}$ act on the input system only.
Then, using Lemma \ref{lem:non-commuting-codes}, we show that it is a convex combination of stabiliser operations where the measurements are given by mutually orthogonal stabiliser code projectors.
To this end, we consider consecutive measurements along a branch and argue again via induction over the depth $k$ of the tree.
Assume that up to depth $k-1$, all consecutive measurements are mutually commuting.
Let $X_{k-1}\subset\F_d$ be the set of outcomes $x_{[k-1]}$ such that the Pauli measurement given by $P(x_k|x_{[k-1]})$ is not commuting with a previous measurement, say $P(x_t|x_{[t-1]})$ for $t < k$.
Since by assumption the previous measurements mutually commute, we can write using Lemma \ref{lem:non-commuting-codes}
\begin{myalign}
 P(x_k\,|\,x_{[k-1]}) & P(x_{k-1}\,|\,x_{[k-2]})  \cdots P(x_t|x_{[t-1]}) \cdots P(x_1) \\
 &= P(x_k\,|\,x_{[k-1]}) P(x_t|x_{[t-1]}) P(x_{k-1}\,|\,x_{[k-2]}) \cdots   P(x_1) \\
 &= d^{-1/2} V(x_{[k]}) P(x_t|x_{[t-1]}) P(x_{k-1}\,|\,x_{[k-2]}) \cdots   P(x_1) \\
 &= d^{-1/2} V(x_{[k]}) P(x_{k-1}\,|\,x_{[k-2]})  \cdots P(x_t|x_{[t-1]}) \cdots P(x_1),
\end{myalign}
for suitable Clifford unitaries $V(x_{[k]})$.
Note that the remaining projectors are mutually commuting by assumption.
As before, this implies that for all branches starting with $x_{[k-1]}\in X_{k-1}$, we can treat $x_k\equiv y$ as the outcome of a classical, uniformly distributed random variable $Y$ and obtain new stabiliser operations $\mathcal{O}(y)$ by conditioning on its outcomes $y\in\F_d$:
\begin{align}
  K'(x_l,\dots,x_{k+1},y,x_{[k-1]}) &:= d^{1/2} K(x_l,\dots,x_{k+1},y,x_{[k-1]}), & x_{[k-1]} \in X_{k-1}, \\
  K'(x_l,\dots,x_{k},x_{[k-1]}) &:= K(x_l,\dots,x_{k},x_{[k-1]}), & x_{[k-1]} \notin X_{k-1}.
\end{align}
As in the previous argument, we have $\mathcal{O}= d^{-1} \sum_y \mathcal{O}(y)$ and proceeding with the induction for all $\mathcal{O}(y)$ shows that $\mathcal{O}$ can be written as a convex combination of stabiliser operations involving only mutually commuting, consecutive measurements.

Combining the above arguments shows that the initial stabiliser operation $\mathcal{O}$ defined in Eq.~\eqref{eq:prop1-kraus-operators} is a convex combination of stabiliser operations $\mathcal{O}'$ where all consecutive measurements are mutually commuting and not acting on ancilla qudits.
This implies that the mutually commuting projectors in every branch $i$ of the SO $\mathcal{O}'$ define a stabiliser code projector $P_i$ and trace-preservation requires that $\sum_i P_i = \one$.
Taking the trace inner product with some $P_j$ shows that this can only be fulfilled if the projectors are mutually orthogonal.
Hence, the terms in the convex combination are of the required type \eqref{eq:SO-syndrome-measurement}.
The additional use of classical randomness simply allows for arbitrary convex combinations of stabiliser operations of type \eqref{eq:SO-syndrome-measurement}.
\end{proof}

%---------------------------------------------------------
\subsection{Separation of CSP and SO in higher dimensions}
\label{sec:csp-neq-so}
%---------------------------------------------------------

To prove our main Theorem \ref{thm:main}, we proceed by generalising the results of Sec.~\ref{sec:min-result} on almost-diagonal channels and derive constraints on almost-diagonal stabiliser operations using the normal form in Thm.~\ref{thm:kraus-decomposition-SO}.

As in Sec.~\ref{sec:min-result}, we define the convex set of \emph{almost-diagonal channels} $\AD_n(d)$ as the set of quantum channels $\mathcal{E}:\,L((\C^d)^{\otimes n})\rightarrow L((\C^d)^{\otimes n})$ that act on the computational basis in the following way:
\begin{align}
    \mathcal{E}(\ketbra{0}{0}) &=  \ketbra{+}{+}, &
			\mathcal{E}(\ketbra{x}{x}) &=  \ketbra{x}{x} \quad x\in\F_d^n\setminus 0,
    \label{eq:def-AD-2} 
\end{align}
where we denote $\ket+ := d^{-n/2}\sum_{x\in\F_d^n}\ket{x}$.
A high-level reason why $\AD_n(d)$ might be relevant for the separation of $\CSP_n(d)$ and $\SO_n(d)$ is given by the observation that $\AD_n(d)$ defines a \emph{face} of the convex set of quantum channels.
In particular, $\AD_n(d) \cap \CSP_n(d)$ is a face of the CSP polytope and thus lies in its boundary.
To see this, consider the linear functional on quantum channels,
\begin{align}
    L(\mathcal{E}) := \frac{1}{d^n}\bigg( \sandwich{+}{\mathcal{E}( \ketbra{0}{0})}{+} + \sum_{x\neq 0} \sandwich{x}{\mathcal{E}(\ketbra{x}{x})}{x} \bigg)  
    \le \frac{1}{d^n} (1+ (d^n-1) \cdot 1)  
    = 1,
\end{align}
with equality if and only if $\mathcal{E} $ satisfies Eq.~\eqref{eq:def-AD-2}.
This shows that $\AD_n(d)$ is the intersection of a supporting hyperplane with the set of quantum channels, in particular it is a face.

As in the case $d=n=2$, the subpolytope $\AD_n(d) \cap \CSP_n(d)$ of $\CSP_n(d)$ is isomorphic to a subpolytope $\mathrm{P}_{n}(d)$ of the $n$-qudit stabiliser polytope which we define in the following.	

\begin{definition}
Let $\mathrm{P}_{n}(d)$ be the polytope of matrices $\sigma$ such that (1) $\sigma$ is a convex combination of $n$-qudit stabiliser states orthogonal to $\ket 0$, and (2)  the diagonal entries are $\sigma_{x,x} = (d^n-1)^{-1} \delta_{x\neq 0}$.
\end{definition}

\begin{lemma}\label{lem:diag-subpolytope-2}
Let $\mathcal{E}\in\AD_n(d)$ be an almost-diagonal quantum channel on $n$ qudits.
Then $\mathcal{E}$ is CSP if and only if it is of the form
\begin{align}
    \mathcal{E}(\rho) = (d^n-1)\,\sigma \circ \rho + \langle 0|\rho|0\rangle \, |+\rangle\langle+|,
\end{align}
where $\circ$ denotes the \emph{Hadamard} (or element-wise) product of two matrices and $\sigma \in \mathrm{P}_{n}(d)$.
In particular, the polytopes $\AD_n(d)$ and $\mathrm{P}_n(d)$ are isomorphic.
\end{lemma}

The proof of Lemma \ref{lem:diag-subpolytope-2} is analogous to the case $d=n=2$, i.e.~Lemma \ref{lem:diag-subpolytope}, and is thus omitted.
Lemma \ref{lem:diag-subpolytope-2} implies that any CSP map fulfilling the constraints \eqref{eq:def-AD-2} is indeed ``almost-diagonal'' in the computational basis in the sense that $\mathcal{E}(\ketbra{x}{y}) \propto \ketbra{x}{y}$ except for $x=y=0$.
Hence, the matrix representation of $\mathcal{E}$ is a diagonal matrix with the first column (corresponding to $x=y=0$) replaced by $(d^{-n},\dots,d^{-n})^\top$.

\begin{lemma}
\label{lem:so-in-ad}
 Any stabiliser operation in the subpolytope $\AD_n(d) \cap \CSP_n(d)$
 is in the convex hull of stabiliser operations $\mathcal{O}$ defined through Lemma \ref{lem:diag-subpolytope-2} by mixed stabiliser states
 \begin{equation}\label{eq:extremal-Operations-In-ADn}
  \sigma = \frac{1}{d^n-1} \sum_{K\in \mathcal{K}} |K| \ketbra{s_K}{s_K} \in\mathrm{P}_n,
 \end{equation}
 where $\mathcal{K}$ is a disjoint partition of $\F_d^n\setminus 0$ by affine spaces $K\subset\F_d^n$ and $\ket{s_K}$ are stabiliser states supported on $K$.
\end{lemma}

\begin{remark}
	Note that not every $ \sigma $ of the form \eqref{eq:extremal-Operations-In-ADn} gives rise to a stabiliser operation $ \mathcal{E} \in \AD_n(d)\cap\SO_n(d) $. For stabiliser operations, only particular partitions $ \mathcal{K} $ are allowed. These partitions exhibit a certain \emph{tree structure}, as explained in Proposition \ref{prop:polar-form} and \cite{campbell_structure_2009}. 
	
	Moreover, not all such $ \sigma $ are extremal within the polytope $\mathrm{P}_n(d)$. For example, if a stabiliser operation contains the measurement of a Pauli operator and the measurement is \emph{not} followed by an operation that is conditioned on at least one of the measurement outcomes, then such an operation cannot be extremal. 
	In this case, the measurement can be replaced by a convex combination of Clifford unitaries. This is a consequence of  Lemma \ref{lem:NoContOperations-NoExtremality} in App.~\ref{app:Measurements that are not followed by adaptive operations are never extremal}.
\end{remark}

To prove Lemma \ref{lem:so-in-ad}, we make use of the following Lemma which allows us to discard ancillary qubits for stabiliser operations in $\AD_n(d)$.

\begin{lemma}
\label{lem:elimination-conditional-cliffords}
 Assume $U\in\Cl_{n+k}(d)$ acts as $U(\ketn{0^k}\otimes\ket{x}) = c_x \ket{s_x}\otimes\ket{x}$ for $c_x\in\C$, some $k$-qudit stabiliser state $\ket{s_x}$ and $x$ is taking values in a subset $K\subset\F_d^n$.
 Then, there exists a diagonal Clifford unitary $D\in\Cl_n(d)$ and a subspace $M\subset \F_d^n$ such that the following identity holds for all $x,y\in K$:
 \begin{equation}
 \label{eq:elimination-conditional-cliffords}
  \tr_{1,\dots,k} \left( U \ketbran{0^k}{0^k}\otimes\ketbra{x}{y} U^\dagger \right) = D \left( \sum_{j=1}^{|M|} Q_{j}\ketbra{x}{y} Q_{j}  \right) D^\dagger.
 \end{equation} 
 Here, $Q_{j}$ are the mutually orthogonal projectors onto the joint eigenspaces of $Z(z)$ for $z\in M$.    
\end{lemma}

As we do not use this formulation in the following, we leave it to the reader to verify that the right hand side of Eq.~\eqref{eq:elimination-conditional-cliffords} can also be written as
\begin{align}
 \sum_{j=1}^{|M|} Q_{j}\ketbra{x}{y} Q_{j} 
 = \frac{1}{|M|} \sum_{z \in M} Z(z) \ketbra{x}{y} Z(z)^\dagger, \qquad \forall x,y\in\F_d^n.
\end{align}

\begin{proof}
The stabilisers of the states $\ket{s_x}$ can only differ by a character and hence we can find a Clifford $V\in\Cl_k(d)$ on the first system such that $V\ket{s_x}=\ket{f(x)}$ with $ f(x) \in \F_d^n $ for all $x\in K$.
Moreover, we find using the cyclicity of the partial trace:
\begin{align}
\label{eq:cyclicity-partial-trace}
  \tr_{1,\dots,k} \left( U \ketbran{0^k}{0^k}\otimes\ketbra{x}{y} U^\dagger \right) 
  &= \tr_{1,\dots,k} \left( (V\otimes\one)U \ketbran{0^k}{0^k}\otimes\ketbra{x}{y} U^\dagger (V^\dagger\otimes\one)\right).
\end{align}
Hence, we may without loss of generality assume that $U(\ketn{0^k}\otimes\ket{x}) = c_x \ket{f(x)}\otimes\ket{x}$ for a suitable function $f$ on $K\subset\F_d^k$.
It is well-known that the Clifford subgroup which normalises the group of Pauli $Z$ operators is given as the semi-direct product of diagonal Clifford unitaries and $CX$ circuits (this follows for instance from the properties of the associated "Siegel parabolic subgroup" of the symplectic group $\Sp_{2n}(\F_2)$, see e.g.~Ref.~\cite{heinrich_2021}).
Thus, the only Clifford unitaries which map computational basis states to computational basis states up to phases are given by this normaliser and $X$ gates. 
Since the second system is fixed for all $x\in K$, we can assume that the $X$ gates act on the first system only and can thus be discarded using the cyclicity of the partial trace, cp.~Eq.~\eqref{eq:cyclicity-partial-trace}.
Then, the form of diagonal Cliffords (see \eg~Ref.~\cite{dehaene_clifford_2003}) implies that $c_x = \sandwich{x}{D}{x}$ for some diagonal Clifford unitary $D\in\Cl_n(d)$.
Moreover, we can find a linear map $F\in\GL_{n+k}(\F_d)$ such that $F(0,x) = (f(x),x)$ and hence, $f$ is linear.

Next, we argue that we can infer whether $\braket{f(y)}{f(x)}$ is zero or one by a suitable measurement on the second system.
To this end, note that this overlap is one exactly if $f(x)=f(y)$, i.e.~$x-y \in \ker f$.
This in turn the case if and only if $z\cdot(x-y) = 0$ for all $z\in M:=(\ker f)^\perp$, hence if and only if$\ket x$ and $\ket y$ lie in the same joint eigenspace of the stabiliser group $\{ Z(z) \; | \; z\in M\}$.
Note that any computational basis state always lies in one of the eigenspaces.
Let $Q_j$ for $j=1,\dots,|M|$ be the projectors on these stabiliser codes.
We thus find
\begin{align}
D \left( \sum_{j=1}^{|M|} Q_{j}\ketbra{x}{y} Q_{j}  \right) D^\dagger
= \bar{c}_y c_x \braket{f(y)}{f(x)} \ketbra{x}{y} 
= \tr_{1,\dots,k} \left( U \ketbran{0^k}{0^k}\otimes\ketbra{x}{y} U^\dagger \right).
\end{align}
\end{proof}

\begin{proof}[Proof of Lemma \ref{lem:so-in-ad}]
  Assume that $\mathcal{O}\in\AD_n(d)$ is a stabiliser operation.
  Without loss of generality, we can assume that $\mathcal{O}$ is extremal in $\SO_n(d)$, since any $\mathcal{O}\in\AD_n(d)\cap\SO_n(d)$ can be written as a convex combination of extremal SO in $\AD_n(d)$ by the same argument as in Lemma \ref{lem:lambda-extremality-AD-qubit}.
  By Proposition \ref{thm:kraus-decomposition-SO}, we can thus assume that $\mathcal{O}$ has the following form
  \begin{equation}
  \label{eq:SO-form-prop1-MH-2}
    \mathcal{O}(\rho) = \tr_{1,\dots,k} \sum_{i=1}^{N} U_i \left(\ketbran{0^k}{0^k}\otimes P_i \rho P_i \right) U_i^\dagger,
  \end{equation}
  where the $P_i$ are mutually orthogonal stabiliser code projectors of rank $d^{n-r_i}$ on the input system, and the $U_i$ are Clifford unitaries conditioned on the measurement outcomes.
  Let $\tilde{\mathcal{O}}\in\SO_{n,n+k}$ be the SO given by Eq.~\eqref{eq:SO-form-prop1-MH-2} without the partial trace.
  Since the defining condition \eqref{eq:def-AD} for $\AD_n$ requires that the reduced state $\tr_{1,\dots,k} \tilde{\mathcal{O}}(\ketbra{x}{x})$ is pure for all $x\in\F_d^n$, it is necessary that $\tilde{\mathcal{O}}(\ketbra{x}{x})=\rho_x\otimes\ketbra{x}{x}$ for $x\neq 0$ and $\tilde{\mathcal{O}}(\ketbra{0}{0})=\rho_0\otimes\ketbra{+}{+}$ else.
  Similar to Eq.~\eqref{eqn:x to x} before, this requires that 
  \begin{align}
   U_i \left( \ketn{0^k}\otimes P_i\ketn{0^n} \right) & \propto \ket{s_{i,0}}\otimes\ket{+}, \label{eq:SO-AD-Kraus-action-1} \\
   U_i \left( \ketn{0^k}\otimes P_i\ket{x} \right) & \propto \ket{s_{i,x}}\otimes\ket{x}, \quad x\neq 0,  \label{eq:SO-AD-Kraus-action-2}    
  \end{align}
  where proportionality can also mean that the RHS vanishes.

  Let $i$ be such that Eq.~\eqref{eq:SO-AD-Kraus-action-1} holds with non-vanishing constant, without loss of generality $i=1$.
  If any of the $P_i$ were non-diagonal, a standard argument (cp.~Lem.~\ref{lem:non-diagonal-projector-stuff} in App.~\ref{sec:misc-stabiliser-stuff}) would show that there exist distinct computational basis states $ \ket{x} \neq \ket{y} $ such that $ P_i\ket{x} = P_i\ket{y} \neq 0 $, which would contradict Eqs.~\eqref{eq:SO-AD-Kraus-action-1} and \eqref{eq:SO-AD-Kraus-action-2}.
  Thus, all $P_i$ are diagonal.
  Moreover, if $P_1$ had rank larger than 1, there would be a $x\neq 0$ such that $P_1\ket{x} = \ket{x}$ is orthogonal to $P_1\ket 0 = \ket{0}$.
  But then, the second factor of $U_1(\ket{0}\otimes P_1\ket{x})$ has to be an $X$ eigenstate, in contradiction to Eq.~\eqref{eq:SO-AD-Kraus-action-2}.
  Hence, the projectors have the form 
  \begin{equation}
  	P_1 = \ketbra{0}{0}, \qquad P_i = \sum_{x\in K_i} \ketbra{x}{x},
  \end{equation}
  where $0 \notin K_i\subset \F_d^n$ is an affine subspace not containing zero.
  Then, orthogonality of the $P_i$ implies that the $K_i$ form a disjoint partition of $\F_d^n\setminus 0$.

  Note that we can assume that $U_1=V\otimes H^{\otimes n}$ (up to $Z$ operators).
  Therefore, we can simply trace out the ancilla for the first term. 
  For $i > 1$, consider the conditional Clifford unitary $U_i$ which acts on the code space $K_i$ as $U_i\ket{0}\otimes\ket{x} = c_i(x) \ket{s_{i,x}}\otimes\ket{x}$ where $c_i(x) \in \C$.
  Then, we can use Lemma \ref{lem:elimination-conditional-cliffords} to replace this action by a diagonal Clifford $D_i$ and $m_i$ mutually orthogonal diagonal stabiliser code projectors on the \emph{input system}.
  We can write these projectors as
  \begin{equation}
   Q^i_j = \sum_{x\in A^i_j} \ketbra{x}{x}, \qquad j=1,\dots,m_i,
  \end{equation}
  where $A^i_j\subset\F_d^n$ are suitable affine subspaces (which might contain zero), forming a disjoint partition of $\F_d^n$.
  Now consider
  \begin{equation}
   P_{i,j}:= Q^i_j P_i = \sum_{x\in A^i_j \cap K_i} \ketbra{x}{x}.
  \end{equation}
  Here, $K^i_j := A^i_j \cap K_i$ is an affine subspace not containing zero.
  Note that $\{K^i_j\}_{j=1,\dots,m_i}$ is a disjoint partition of $K_i$ and thus, the affine subspaces $\mathcal{K}:=\{K^i_j \; | \; i=2,\dots,N, \, j=1,\dots,m_i \}$ obtained in this way form a disjoint partition of $\F_d^n\setminus 0$.
  Finally, we can write
  \begin{equation}
   \mathcal{O}(\rho) 
   = \ketbra{+}{0}\rho\ketbra{0}{+} + \sum_{i=2}^N  \sum_{j=1}^{m_i} D_i Q^i_j P_i \rho P_i Q^i_j D_i^\dagger 
   = \ketbra{+}{0}\rho\ketbra{0}{+} + (d^n-1)\, \sigma \circ \rho,
  \end{equation}
  where
  \begin{equation}\label{eq:Explicit-Form-Of-Sigma-In-AD-cap-SO}
   \sigma = \frac{1}{d^n-1} \sum_{(i,j)} |K^i_j| \, \ketbra{s_{i,j}}{s_{i,j}}, \qquad \ket{s_{i,j}} := |K^i_j|^{-\frac 1 2} \sum_{x \in K^i_j} D_i \ket{x}.
  \end{equation}  
  To see that this is indeed an $\AD_n$ channel, note that any of the $\ket{s_{i,j}}$ is a stabiliser state and as the $K^i_j$ form a disjoint partition of $\F_d^n\setminus 0$, $\sigma$ is a proper convex combination and hence in $\SP_n(d)$.
  Finally, we have $\sandwich{0}{\sigma}{x} = 0$ for all $x\in\F_d^n$ and for $x\neq 0$:
  \begin{equation}
    (d^n-1)\sandwich{x}{\sigma}{x} 
    = \sum_{(i,j)} |K^i_j| \, |\braket{x}{s_{i,j}}|^2 = \sum_{(i,j)} \ind_{K^i_j}(x).
  \end{equation}
  Since any $x\neq 0$ is in exactly one affine subspace $K^i_j$, we thus find $\sigma\in\mathrm{P}_n(d)$.
\end{proof}

In Section \ref{sec:min-result}, Lemma \ref{lem:lambda-extremality-AD} has been proven for the case $n=d=2$.
Here, we treat the general case.

\begin{lemma}
\label{lem:lambda-extremality-AD}
The matrix $\lambda$ with elements
\begin{align*}
    \lambda_{x, t x} &= \lambda_{tx,x} =(d^n-1)^{-1} \delta_{t, 1}, 
    \quad \forall x \in \FF_d^n, t \in \F_d,
    &
    \lambda_{x,y} &= d^{-1}(d^n-1)^{-1},
    \quad \forall\,0\neq x \neq y \neq 0
\end{align*}
is the unique maximiser in $\mathrm{P}_n$ of the linear function $L: \sigma \mapsto \sandwich{+}{\sigma}{+}$ with $L(\lambda)=1/d$.
In particular, $\lambda$ is a vertex of $\mathrm{P}_{n}$.
\end{lemma}

\begin{proof}[Proof of Lemma \ref{lem:lambda-extremality-AD}]
For any $\sigma\in\mathrm{P}_n$, we have
\begin{align}
\label{eq:L-functional}
    L(\sigma)
    =
    \langle +| \sigma | + \rangle
    =
    \sum_s p_s |\langle +| s\rangle|^2,
\end{align}
where $\ket s$ ranges over stabiliser states orthogonal to $\ket{0}$.
Among those, the inner product $|\braket{+}{s}|^2$ is maximal and equal to $1/d$ exactly for states such that  $\braket{x}{s}$ is proportional to an indicator function on an affine space $K$ of codimension $1$ with $0\neq K$.
We call $K$ a \emph{proper affine hyperplane}.

We can write any affine hyperplane as $K=V+w$ where $V$ is a linear subspace of codimension 1 and $w$ is an affine shift.
Those are only determined modulo $V$, hence there are $|\F_d^n/V|= d$ many.
The condition $0\neq K$ implies that the shift cannot be trivial, eliminating one possibility.
The number of linear subspaces of codimension 1 is given by the Gaussian binomial coefficient $\binom{n}{n-1}_d = \frac{1-d^n}{1-d}$, hence the number of proper affine hyperplanes is $(d-1)\frac{1-d^n}{1-d} = d^n-1$.

Define $\lambda$ to be the uniform convex combination of all maximising stabiliser states $\ket{K}$ given by indicator functions on the proper affine hyperplanes $K$.
For any $x\in\F_d^n\setminus 0$, the diagonal entry $\lambda_{x,x}$ is the overlap $\braket{x}{K}\propto \ind_K(x)$ averaged over $K$.
As $\GL(\F_d^n)$ acts transitively on both the non-zero points in $\F_d^n$ and the affine spaces not containing zero, $\lambda_{x,x}$ cannot depend on $x\neq 0$.
Since $\tr \lambda = 1$ and $\lambda_{0,0}=0$, we thus find $\lambda_{x,x}=(d^{n}-1)^{-1}$.
For the off-diagonal entries $\lambda_{x,y}$ with $x\neq y$, we can argue similarly.
First, as no $K$ contains zero, we have $\lambda_{x,0}=\lambda_{0,x}=0$.
Moreover, if $x\in K$, no non-trivial multiple of $x$ is in $K$, thus $\lambda_{x,tx}=\lambda_{tx,x}=0$ for $t\neq 1$.
In any other case, $\{x,y\}$ is linearly independent and transitivity again implies that $\lambda_{x,y}$ cannot depend on $(x,y)$.
There are in total $(d^n-1)(d^n-d)$ many of these pairs.
By construction, $L(\lambda)=1/d$, and writing out this condition then yields $\lambda_{x,y} = d^{-1}(d^n-1)^{-1}$.

It remains to be shown that this solution is unique.
To this end, we claim that the $d^n-1$ indicator functions $\ind_K$ on the proper affine hyperplanes $K$ are linearly independent.
As the main diagonal of the density matrix of any stabiliser state $\ket{K}$ is proportional to $\ind_K$, the $d^n-1$ constraints $\sigma_{x,x}=(d^n-1)^{-1}$ for $x \neq 0$ defining $\mathrm{P}_{n}$ then single out the above constructed $\lambda$.

To prove the claim, we apply the standard (cyclic) Fourier transform on $\FF_d^n$. 
Clearly, the set of indicator functions on proper affine hyperplanes $\{\ind_K\}$ is linearly independent if and only if their Fourier transforms are.
The image of the indicator function on $K=V+w$ is a function with support on $V^\perp$ and values proportional to the additive character 
\begin{align*}
    \chi: V^\perp \to \CC, \qquad x\mapsto \omega^{w \cdot x}.
\end{align*}
As we assume $w$ to be non-trivial, varying $w$ results in the set of non-trivial characters on the one-dimensional subspace $V^\perp$.
Thus, the set $\{\ind_K\}$ maps to the set of non-trivial characters on the one-dimensional subspaces of $\F_d^n$.
On a fixed subspace $V^\perp$, the non-trivial characters are linearly independent and this is still true for their restriction to the non-zero points $V^\perp\setminus 0$.
Since the non-zero points of one-dimensional subspaces form a disjoint partition of $\FF_d^n\setminus 0$, the set of all non-trivial characters of one-dimensional subspaces is also linearly independent.
\end{proof}

From the proof and Lemma \ref{lem:diag-subpolytope-2} it is clear that the matrix $\lambda$ defines a CSP channel which should be understood as a generalisation of the $\Lambda$ channel given for $n=d=2$ in Sec.~\ref{sec:min-result}.
For qubits, $d=2$, this channel reads as follows:
\begin{equation}
   \Lambda(\rho) := \rho_{00} \ketbra{+}{+} +  \sum_{x\in\F_2^n\setminus 0} \rho_{xx} \ketbra{x}{x} + \frac 1 2 \sum_{\substack{x,y\in\F_2^n\setminus 0 \\ x \neq y}} \rho_{xy} \ketbra{x}{y}, \qquad \rho_{xy}:=\sandwich{x}{\rho}{y}.
\end{equation}

The $n\geq 2$ case in our main Theorem \ref{thm:main} now follows from the straightforward observation that the linear functional $L$ is always strictly less than its maximum $1/d$ on elements of the form given in Lemma \ref{lem:so-in-ad}, in particular on stabiliser operations.

\begin{corollary}[$\SO_n\cap\AD_n \neq \CSP_n \cap \AD_n$]
\label{cor:SO-cap-AD}
 For $n\geq 2$, the value of the linear functional $L$ on $\SO_n\cap\AD_n$ is strictly smaller than $1/d = L(\lambda)$. 
 In particular, $\SO_n\cap\AD_n \neq \CSP_n \cap \AD_n$.
\end{corollary}

\begin{proof}
Consider a stabiliser operation with $ \sigma \in P_n $ as constructed in Lemma \ref{lem:so-in-ad}, i.e. 
\begin{equation}
	\sigma = \frac{1}{d^n-1} \sum_{K\in \mathcal{K}} |K| \ketbra{s_K}{s_K}.
\end{equation}
Since $n\geq 2$, the disjoint partition $ \mathcal{K} $ of $ \FF_d^{n}\setminus \{0\} $ cannot only contain affine subspaces $ K $ of codimension $ 1 $, so $ |\langle + | s_K \rangle |^2 \le  1/d $ and $ |\langle + | s_K \rangle |^2 <  1/d  $ for at least one $ K \in \mathcal{K} $.
Hence, 
\begin{align}\label{eq:estimation-of-overlapI}
	\langle + | \sigma | + \rangle = \frac{1}{d^n-1}  \sum_{K \in \mathcal{K}} |K| \, | \langle +|s_K \rangle |^2  < \frac{1}{d^n-1}\sum_{K\in \mathcal{K}} |K| \frac{1}{d} = \frac{1}{d}.
\end{align}
\end{proof}

\begin{remark}
	To get a quantitative statement about the separation of $ \SO_n $ and $ \CSP_n $ within the polytope $ \AD_n $, we give a upper bound for 
	\begin{myalign}
		\max_{\sigma \in \SO_n \cap \AD_n}\langle +| \sigma | +\rangle.
	\end{myalign}
As in Eq.~\ref{eq:estimation-of-overlapI}, we have 
\begin{align}\label{eq:estimation-of-overlapII}
	\langle + | \sigma | + \rangle = \frac{1}{d^n-1}  \sum_{K \in \mathcal{K}} |K| \, | \langle +|s_K \rangle |^2  \le \frac{1}{d^n-1} \sum_{K \in \mathcal{K}} |K|\frac{|K|}{d^n} = \frac{1}{d^n(d^n-1)} \sum_{K \in \mathcal{K}} |K|^2.
\end{align}
The RHS gets large when affine subspaces in the disjoint partition $ \mathcal{K} $ of $ \F_2^n \setminus \{0\} $ have large cardinality. However, the partition $ \mathcal{K} $ must be chosen according to Thm.~\ref{thm:kraus-decomposition-SO}. Thus, all projectors $ \ketbra{s_K}{s_K} $ belong to  a projective measurement, which also contains the measurement of $ \ketbra{0^n}{0^n} $, due to the proof of Lemma \ref{lem:so-in-ad}. 
We conjecture that such a projective measurement which maximizes \eqref{eq:estimation-of-overlapII}
has the following form:
\begin{enumerate}[label=(\roman*)]
	\item Measure the first qudit in the computational basis.   
	\item If $ x \neq 0 $ is measured, do nothing.  
    If $ 0 $ is measured, measure the second qudit in the computational basis.  
    Continue in this fashion until all qubits are measured. 
	\item If $ 0 $ has been measured on every qudit, apply a Hadamard gate to every qudit.  
\end{enumerate}
Then, $ \ket{s_K} = \sum_{x \in K} \ket{x} $ and every $ K $ is of the form 
\begin{align}
	K = xe_i + (0^i \oplus \FF_d^{n-i }) \qquad \text{with} \qquad x \in \{1,\ldots,d-1\}, \quad   |K| = d^{n-i},
\end{align}
where $ e_i $ is the $ i $-th standard basis vector for $ i = 0,.\ldots,n-1 $. 
Thus, we have
\begin{align}
	\sum_{K \in \mathcal{K}} |K|^2 = \sum_{k = 0}^{n-1}  (d-1) d^{2k} = (d-1) \frac{d^{2n-2}}{d^2-1},
\end{align}
where we used that the expression is a geometric sum. 
Hence, 
\begin{multline}
	\langle + | \sigma | + \rangle  = \sum_{K \in \mathcal{K}} |K|\frac{|K|}{d^n} = \frac{1}{d^n(d^n-1)} \sum_{K \in \mathcal{K}} |K|^2\\
     = \frac{(d-1)d^{2n-2}}{d^n(d^n-1)(d^2-1)} = \frac{(d-1)d^{n-2}}{(d^n-1)(d^2-1)} \approx \frac{1}{d^3}.
\end{multline}
As the above stabiliser operation gives a upper bound on $ \max_{\sigma \in \SO_n \cap \AD_n}\langle +| \sigma | +\rangle $, this shows a separation between $\SO_n\cap\AD_n$ and $\CSP_n\cap\AD_n$
which depends, however, only on $ d $ and not on $n$.
\end{remark}

%-----------------------------------------------------------
\subsection{Equality of SO and CSP in the single-qudit case}
\label{sec:csp-eq-so}
%-----------------------------------------------------------
In this section, we will prove that $ \CSP $-channels coincide with stabiliser operations in the single-qudit case. More precisely, we will show that every extremal $ \CSP $-map is a Pauli measurement followed Clifford unitaries conditioned on the possible measurement outcomes. 
In the proof we will make use of the polar form of $ \CSP $-maps, Eq.~\eqref{eq:csp-polar-decomposition}, App.~\ref{sec:additional} (for more details, see \cite{heinrich_2021}).

To prove the statement, we will use the following auxiliary lemma:
\begin{lemma}
\label{lem:double-projectors}
Suppose $\mathcal{E}\in\CSP_n$ is a CSP map given in the polar form \eqref{eq:csp-polar-decomposition}
\begin{myalign}
\mathcal{E} = \sum_{i} c_i U_i P_i \cdot P_i U_i^\dagger, \qquad \text{where } c_i > 0 \text{ for all } i.
\end{myalign}
Assume that there exists an index pair $(k,\ell)$ with $P_{k}=P_{\ell}$ but $U_{k} P_{k }\neq U_{\ell} P_{\ell}$.
Then, $\mathcal{E}$ is not extremal.
\end{lemma}

\begin{proof}
Since $ \mathcal{E} \in \CSP_n$, the projectors $ P_i $ satisfy the TP-condition \eqref{eq:povm-tp-cond}
\begin{align}
	\one = c_k P_k + c_\ell P_\ell + \sum_{k \neq i \neq \ell} c_i P_i 
\end{align} 
and therefore 
\begin{align}
	\one = (c_k+c_\ell)P_k + \sum_{k \neq i \neq \ell} c_i P_i \quad \text{and} \quad  \one = (c_k+c_\ell) P_\ell + \sum_{k \neq i \neq \ell} c_i P_i.
\end{align}
Hence, $ \mathcal{E} $ is a convex combination $ \mathcal{E} = \frac{c_k}{c_k+c_\ell} \mathcal{E}_k + \frac{c_\ell}{c_k+c_\ell} \mathcal{E}_\ell $ of the two distinct CSP-channels
\begin{myalign}
	\mathcal{E}_k &= (c_k+c_\ell)  U_kP_k \cdot P_k U_k + \sum_{k \neq i \neq \ell} c_iU_iP_i \cdot P_i U_i^\dagger,  \\
	\mathcal{E}_\ell &=   (c_k+c_\ell)  U_\ell P_\ell \cdot P_\ell U_\ell + \sum_{k \neq i \neq \ell} c_iU_iP_i \cdot P_i U_i^\dagger, 
\end{myalign}
so $ \mathcal{E} $ cannot be extremal. 
\end{proof}

\begin{theorem}
Let $ \mathcal{E} \in \CSP_1$ be an extremal CSP map on a single qudit of prime dimension $d$. 
Then, either $ \mathcal{E} = U\cdot U^\dagger $ for some Clifford unitary $ U $ or $ \mathcal{E} $ is of the form 
\begin{myalign}\label{eq:extremal-channel-n-1}
	\mathcal{E} = \sum_{i=1}^d U_i P_i \cdot P_i U_i^\dagger,
\end{myalign}
where $ \{ P_i \}$ are the $d$ mutually orthogonal stabiliser code projectors associated to the eigenspaces of a Pauli operator and $ \{ U_i \} $ are Clifford unitaries.
Since such a channel $ \mathcal{E} $ can be realised via stabiliser operations, it follows $  \SO_1 = \CSP_1 $.
\end{theorem}

\begin{proof}
Using the characterization of completely stabiliser preserving maps, cf.~Eq.~\eqref{eq:csp-polar-decomposition}, we may assume that a $ 1 $-qudit CSP channel is of the form
\begin{align}
	\mathcal{E} = d\sum_{i} \lambda_i U_i P_i  \cdot P_i U_i^\dagger + \sum_{j} \hat{\lambda}_j V_j \cdot V_j^\dagger
\end{align}
for coefficients $ \lambda_i, \hat{\lambda}_j \ge 0$ 
with $ \sum_i \lambda_i + \sum_j \hat{\lambda}_j = 1$, 
 Clifford unitaries $ U_i, V_j $ and stabiliser code projectors $ P_i $ which satisfy the TP-condition \eqref{eq:povm-tp-cond} :
\begin{align}\label{eq:TP-condition-n=1}
	\one = \mathcal{E}^\dagger(\one) = d \sum_{i} \lambda_i P_i. 
\end{align}
Since any channel that simply conjugates the input with a Clifford unitary $ U $ is already an extremal CSP channel, $ \mathcal{E} $ can only be extremal
if	(1)  there is exactly one non-zero $ \hat{\lambda}_j $ and $ \lambda_i = 0 $ for all $ i $ (which means that $ \mathcal{E}= U \cdot U^\dagger $ for some Clifford unitary $ U $), or (2)  $ \hat{\lambda}_j = 0 $ for all $ j $.

In the second case, note that for $n=1$, all stabiliser projectors have rank 1 and project onto a stabiliser state.
There are in total $d(d+1)$ stabiliser states $\ket{\phi_{i,a}}$ which form a complete set of mutually unbiased bases, in particular for any $a=1,\dots,d+1$ the set $\{\ket{\phi_{i,a}}\}_i$ is an orthonormal basis.
%and  $|\braket{\phi_{i,a}}{\phi_{j,b}}|^2 = d^{-1}$ whenever $a\neq b$.
By Lemma \ref{lem:double-projectors}, we can assume that every projector $\ketbra{\phi_{i,a}}{\phi_{i,a}}$ only occurs at most once in $\mathcal{E}$.
Thus, grouping the projectors by their basis, we can write the CSP channel $\mathcal{E}$ as
\begin{align}
	\mathcal{E} = d \sum_{a=1}^{d+1} \sum_{i=1}^d \lambda_{i,a} U_{i,a} \ketbra{\phi_{i,a}}{\phi_{i,a}} \cdot \ketbra{\phi_{i,a}}{\phi_{i,a}} U_{i,a}^\dagger.
\end{align} 
Since every basis $\{\ket{\phi_{i,a}}\}_i$ is the eigenbasis of a (non-trivial) Pauli operator $ w $ and $ \Tr(w \ketbra{\phi_{i,a'}}{\phi_{i,a'}}) = 0 $ for $ a \neq a' $, taking the trace inner product of Eq.~\eqref{eq:TP-condition-n=1} multiplied with $ w $ implies 
\begin{equation}
  0 = \sum_{i=1}^d \lambda_{i,a} \omega^i,
\end{equation}
which forces the $\lambda_{i,a}$ to be either identically zero or independent of $i$. 
Setting $\tilde\lambda_a = d^{-1}\lambda_{i,a}$, we thus arrive at a convex combination of $\mathcal{E}$ into CSP channels $\mathcal{E}_a$:
\begin{align}
	\mathcal{E} = \sum_{a=1}^{d+1} \tilde\lambda_a \mathcal{E}_a, \qquad \mathcal{E}_a := \sum_{i=1}^d U_{i,a} \ketbra{\phi_{i,a}}{\phi_{i,a}} \cdot \ketbra{\phi_{i,a}}{\phi_{i,a}} U_{i,a}^\dagger.
\end{align} 
Hence, extremality of $\mathcal{E}$ implies that it is of the desired form.
\end{proof}

%%% =============================================
\section{Additional properties of CSP channels and examples}
\label{sec:additional}
%%% =============================================

In this section, we derive additional properties of completely stabiliser-preserving channels which are not directly used to show the main result of this paper.
We characterise CSP channels in terms of certain generalised stabiliser measurements and adaptive Clifford unitaries.
This is what we call the \emph{polar form} and has been used in Sec.~\ref{sec:csp-eq-so} as well as in the simulation protocol of \textcite{seddon_quantifying_2021}.
We then use this characterisation to compile a list of examples of CSP channels.

By Lemma \ref{lem:csp-choi}, completely stabiliser-preserving maps are in bijection with the subset of the bipartite $2n$-qudit stabiliser polytope fulfilling the TP condition. 
Notably, bipartite stabiliser states have a special structure that can be exploited to bring them into a standard form which we call the \emph{polar form}. 
It is given by $\ket{s}= d^{k/2} U P\otimes \one \ketb{\phi^+}$ for a Clifford unitary $U\in\Cl_{n}(d)$ and a stabiliser code projector $P$ of rank $d^{n-k}$. 
Note that from this form, one can immediately derive the Schmidt rank of $\ket{s}$ as $\log_d\rank(P) = n-k$.
While this fact seems to be folk knowledge in the relevant community and related results can be found in Refs.~\cite{howe_1973,howe_1988,fattal_entanglement_2004}, we have been unable to find an explicit formulation in the literature.
A proof of this fact can be found in the PhD thesis of one of the authors \cite[Sec.~12.3.2]{heinrich_2021}.

\begin{proposition}
	\label{prop:polar-form}
	The $2n$-qudit state $\ket{s}\in(\C^d)^{\otimes 2n}$ is a stabiliser state if and only if there is a Clifford unitary $U\in\Cl_n(d)$ and a stabiliser code $P\in\stab(k,n)$ such that $\ket{s}=d^{k/2} U P\otimes \one \ketb{\phi^+}$.
\end{proposition}

While the projective part in the polar form of a stabiliser state is unique, the unitary part is not. 
This is because replacing the Clifford unitary by $U\mapsto UV$ where $V$ acts trivially on the code space gives an equivalent presentation of the state.
Technically, this means that the unitary part is unique \emph{up to the left Clifford stabiliser} of the stabiliser code.

Using the polar form, the polytope of CSP maps can be characterised as follows:
The $\SP_{2n}(d)$ polytope corresponds under the inverse Choi-Jamiołkowski isomorphism to the polytope which is spanned by channels with a single stabiliser Kraus operator $d^{k/2}UP$.
Hence, any CSP map is of the form 
\begin{equation}
	\label{eq:csp-polar-decomposition}
	\mathcal E = \sum_{i=1}^r \lambda_i \frac{d^n}{\rank P_i} U_i P_i \cdot P_i U_i^\dagger, 
\end{equation}
where the $\lambda_i$ form a probability distribution.
However, Eq.~\eqref{eq:csp-polar-decomposition} only defines a valid CSP map $\mathcal{E}$ if it is trace-preserving.
We can cast the TP condition into an appealing form: 
$\mathcal E$ is a CSP map if and only if in addition to Eq.~\eqref{eq:csp-polar-decomposition}, it fulfils
\begin{equation}
	\label{eq:povm-tp-cond}
	\one = \mathcal E^\dagger(\one) =  \sum_{i=1}^{r}  \frac{d^{n}\lambda_i}{\rank P_i} P_i.
\end{equation}
Thus, a sufficient and necessary condition for a convex combination of stabiliser Kraus operators to define a CSP map is that the rescaled projective parts $\tilde P_i := (d^n\lambda_i/\rank P_i) P_i$ form a POVM. 
In this context, the CSP channel $\mathcal{E}$ in Eq.~\eqref{eq:csp-polar-decomposition} can be seen as the quantum instrument associated with the stabiliser POVM $\{\tilde P_i\}$ combined with the application of Clifford unitaries $U_i$ conditioned on outcome $i$.

A possible solution to Eq.~\eqref{eq:povm-tp-cond} is a \emph{syndrome measurement}, i.e.~the POVM that is defined by the measurement of a set of mutually commuting Pauli operators (cp.~example 3 below).
Then, the corresponding CSP channel $\mathcal E$ is a stabiliser operation.
However, as stabiliser operations are also allowed to use auxiliary qubits, they can effectively induce more complicated POVMs that fulfil Eq.~\eqref{eq:povm-tp-cond}.
A priori, it is thus not clear whether CSP channels are different from stabiliser operations (this is, of course, answered by our main theorem \ref{thm:main}).
Interestingly, it even seems to be difficult to find solutions to Eq.~\eqref{eq:povm-tp-cond} in terms of admissible stabiliser codes $P_i$ and coefficients $\lambda_i$. 
In particular, one could think of arranging overlapping codes with the right weights in non-trivial ways such that they yield the identity on Hilbert space. 
Indeed, an example of a CSP channel defined via overlapping stabiliser codes is the $\Lambda$ channel used for our main argument, see also App.~\ref{sec:on-lambda}.
Note that given a set of stabiliser codes, it is in principle possible to decide whether there exist coefficients such that Eq.~\eqref{eq:povm-tp-cond} holds by solving a linear system of equations which depends on the structure of code overlaps.

Finally, let us give some examples of CSP maps:
\begin{enumerate}
	\item \emph{Mixed Clifford channels.} Take $P_i \equiv \one$, then $d^n/\rank P_i =1$ and Eq.~\eqref{eq:povm-tp-cond} is trivially fulfilled for any convex combination.
	\item \emph{Dephasing in a stabiliser basis.} Take a basis of stabiliser states, and let $P_i$ be the rank-one projectors onto the basis. A uniform convex combination $\lambda_i=d^{-n}$ of these fulfils the TP condition Eq.~\eqref{eq:povm-tp-cond}.
	Such a channel corresponds to a dephasing in the chosen basis, followed by the potential application of conditional Clifford unitaries $U_i$ depending on the basis measurement outcome $i$.
	\item \emph{Dephasing in stabiliser codes.} More generally, take an arbitrary stabiliser group $S=\langle g_1,\dots,g_k\rangle$ and let $P_i$ be all $d^k$ orthogonal stabiliser codes corresponding to different phases of the generators and $\lambda_i = d^{-k}$. 
	This defines a POVM (``syndrome measurement'').
	\item \emph{Reset channels.} Let $s\in\stab(n)$ be an arbitrary stabiliser state and consider the channel which replaces every input by $s$, i.e.~$\mathcal R_s:\, X \mapsto \tr(X) s$. 
	It is clearly CSP and is a special cases of the second example where $\ket s$ is completed to a stabiliser basis and the Clifford unitaries are chosen such that all basis elements are mapped to $\ket s$.
\end{enumerate}

%%% =============================================
\section{Summary and open questions}
\label{sec:summary}
%%% =============================================

In this work, we have studied and compared two classes of free operations in the resource theory of magic state quantum computing, namely completely stabiliser-preserving (CSP) channels and stabiliser operations (SO).
Our main result shows that the set of multi-qudit CSP channels is always strictly larger than its subset of stabiliser operations.
In the single-qudit case, however, the two classes coincide.
Thus, our result is in analogy with the well-known fact from entanglement theory that LOCC operations are contained but not equal to the set of separable quantum channels.

Our proof strategy is simplified by the observation that it is sufficient to show the separation of CSP and SO in a suitable subspace.
Having derived restrictions on the form of CSP and SO channels in this subspace, we then give a linear functional which is able to separate the two sets.
In particular, we explicitly construct a CSP channel which is the unique maximiser of said functional and thus extremal in CSP.

As an auxiliary result, we restrict the form of Kraus operators of extremal stabiliser operations.
In particular, this implies that stabiliser operations can be realised in a finite number of rounds.
This is in contrast to entanglement theory, where the analogous LOCC operations become strictly more powerful with the number of rounds.

\vspace\floatsep

In our operational definition of SO, we intentionally allow for arbitrary classical control logic.
As laid out in Sec.~\ref{sec:stab-ops}, this is implicit in the axiomatic definition of CSP and a separation would otherwise be trivial.
However, as our proof does not depend on the details of the classical control, the separation still holds if we restrict the latter to efficient classical algorithms.
For CSP, this has to be understood in the sense of Sec.~\ref{sec:additional}, i.e.~as efficient classical processing of the outcomes of generalised stabiliser POVMs and control of adaptive Clifford operations.

Some magic monotones, such as the \emph{dyadic negativity} can be connected to classical simulation algorithms \cite{seddon_quantifying_2021}.
These allow to efficiently simulate a restricted class of CSP channels which is, however, strictly smaller than CSP with efficient classical control.
The main reason for this is that it is not clear how to efficiently simulate the generalised stabiliser POVMs introduced in Sec.~\ref{sec:additional}.
Therefore, additional assumptions on these POVMs are necessary.
However, it is plausible that CSP with these restricted POVMs is still strictly larger than SO.
Hence, we expect that the algorithm by \textcite{seddon_quantifying_2021} allows for simulation beyond the Gottesman-Knill theorem.
A thorough analysis of the simulability of CSP channels and comparison with the Gottesman-Knill theorem is left for future work.

\vspace\floatsep

Finally, we think that our result will stimulate further research in the resource theory of magic state quantum computing.
The axiomatic approach to free operations has the advantage that it is possible to directly apply results from general resource theory and obtain explicit bounds on e.g.~state conversion and distillation rates \cite{veitch_resource_2014,liu_one-shot_2019,fang_no-go_2020,seddon_quantifying_2021,wang_efficiently_2020}.
For the case of stabiliser-preserving channels, it is also known that the theory is asymptotically reversible  \cite{liu_many-body_2020}.
Here, it would be interesting to investigate which results still hold when the set of free operations is restricted to CSP.
Moreover, if ``free'' shall have an operational meaning, then the question of simulability and the power of classical control will have to be discussed.

Our separation result opens the possibility that tasks like magic state distillation show a gap in the achievable rates between CSP channels and stabiliser operations.
Again, this question is motivated from entanglement theory, where a significant separation between separable channels and LOCC operations for e.g.~entanglement conversion is known \cite{chitambar_increasing_2012}.

%%% =============================================
\section{Acknowledgements}
%%% =============================================

We would like to thank James Seddon and Earl Campbell for discussions which helped initialising this project.
Furthermore, we thank Felipe Montealegre Mora for many discussions during the various stages of this work, and Mateus Araújo for helpful input on ancilla-assisted operations.
This work has been supported by Germany's Excellence Strategy -- Cluster of Excellence \emph{Matter and Light for Quantum Computing (ML4Q)} EXC2004/1, the Deutsche Forschungsgemeinschaft (DFG, German Research Foundation) within the Emmy Noether program (grant number 441423094) and the Priority Program CoSIP, the German Federal Ministry for Education and Research through the Quantum Technologies program (QuBRA, QSolid) and the German Federal Ministry for Economic Affairs and Climate Action (ProvideQ).
\bibliography{csp}

%%% =============================================
\appendix

%%% =============================================
\section{Miscellaneous facts on stabiliser states}
\label{sec:misc-stabiliser-stuff}
%%% =============================================

Here, we state a fact (Proposition~\ref{prop:well-known stab basis}) on \emph{stabiliser bases} that seems to be widely known, but for which we could not find a direct reference.
It is used in the proof of Lemma~\ref{lem:diag-subpolytope}.

Let $S$ be a stabiliser group on $n$ qudits of size $|S|=d^n$.
There is a unique (up to phases) joint eigenbasis $\{|\alpha_i\rangle\}_i$ of all elements $s$ of $S$.
Concretely:
The eigenvalue equations
\begin{align}\label{eqn:stabiliser basis}
	\chi_i(s) 
	s
	|\alpha_i\rangle 
	= 
	|\alpha_i\rangle,
	\qquad s\in S
\end{align}
establish a one-one correspondence between the set of characters of $S$ and elements of the common eigenvectors.
Bases arising this way are called \emph{stabiliser bases}.

The argument uses basic notions from the description of Pauli operators and stabiliser states in terms of discrete symplectic vector spaces \cite{gross_hudsons_2006}.
In particular, two Pauli operators $w(a), w(b)$ for  commute if and only if the \emph{symplectic inner product} 
\begin{align*}
	[a,b] = 
	\sum_{i=1}^n (a_z)_i (b_x)_i
	-
	\sum_{i=1}^n (a_x)_i (b_z)_i
\end{align*}
is zero (as an element of $\FF_d$).
A subset $M\subset \FF_d^{2n}$ is \emph{isotropic} if the symplectic inner product vanishes between any two elements of $M$.
An isotropic subspace $M$ is \emph{maximal} if $\dim M = n$.
Witt's Lemma implies that every isotropic set is contained in a maximal isotropic subspace.

\begin{lemma}\label{lem:states in pauli span}
	Let $M\subset \FF_d^{2n}$ be an isotropic set.
	There is a stabiliser basis such that all pure states in the linear span of $\{ w(a) \,|\, a \in M\}$ belong to that basis.
\end{lemma}

\begin{proof}
	Choose a basis $b_1, \dots, b_n$ for some maximal isotropic subspace containing $M$.
	Then the operators $w(b_1), \dots, w(b_n)$ generate a stabiliser group $S$ of size $d^n$.
	Their unique common eigenbasis is a stabiliser basis. 
%	By construction, the Pauli operators $w(M)$ commute with $S$ and are thus diagonal in that stabiliser basis, which implies the claim.
	By construction, any pure state $ \ketbran{\psi}{\psi} $ contained in the span of $ \{ w(a) \, | \, a \in M \} $ commutes with the elements in $ S $ and is thus a joint eigenvector of all elements in $ S $, which implies that $ \ket{\psi} $ belongs to the stabiliser basis of $ S $. 
\end{proof}

\begin{proposition}\label{prop:well-known stab basis}
	Let $\ket\psi\in
	({\CC^d})^{\otimes n_1} \otimes 
	({\CC^d})^{\otimes n_2}$
	be a bi-partite stabiliser state.
	Let $\{|\alpha_1\rangle, \dots, |\alpha_{d^{n_1}}\rangle\}$ be a stabiliser basis on the first subsystem.
	Then there is a stabiliser basis on the second subsystem such that each of the
	the partial contractions 
	\begin{align*}
					\ket{\beta_i} = 
					(\bra{\alpha_i}\otimes\one)|\psi\rangle \in (\CC^d)^{\otimes n_2}
	\end{align*}
	is proportional to some element of this basis (with a proportionality constant of $0$ being allowed).
\end{proposition}

\begin{proof}
	Let $S$ be the stabiliser group of $|\alpha_1\rangle$.
	There exists an isotropic subspace $\hat S\subset \FF_d^{2n_1}$ and a function $f_S: \hat S \to \CC$ such that 
	\begin{align*}
		S = \{ f_S(\hat s) \, w(\hat s) \,|\, \hat s \in \hat S \}.
	\end{align*}
	Let $\chi_i$ be the character associated with $|\alpha_i\rangle$ as in (\ref{eqn:stabiliser basis}).
	Analogously, write the stabiliser group of $\ket\psi$ as
	\begin{align*}
		T = \{ f_T(\hat t) \, w(\hat t) \,|\, \hat t \in \hat T \}
	\end{align*}
	for a suitable isotropic subspace $\hat T\subset \FF_d^{2(n_1+n_2)}$ and phase function $f_T: \hat{T} \to \CC$.
	Let 
	\begin{align*}
		\hat U = \{ \hat t_2 \in \FF_d^{2n_2} \,|\, \exists \hat s \in \hat S, \hat s \oplus \hat t_2 \in \hat T \}.
	\end{align*}
	The fact that $\hat S$ and $\hat T$ are isotropic implies that the same is true for $\hat{U}$.

	Using Eq.~(\ref{eqn:stab projection}), we obtain 
	\begin{align*}%\label{eqn:partial}
		\ket{\beta_i}\bra{\beta_i}
		&=
		(\bra{\alpha_i}\otimes\Id)|\psi\rangle\langle\psi| (\ket{\alpha_i}\otimes\Id) \\
		&\propto
		%&=
		%\frac1{|S|\,|T|}
		\sum_{s\in S,t \in T} 
		\chi_i(s) 
		\tr_1 ((s\otimes\Id)  t)  \\
		 &=
		\sum_{\hat s\in \hat S,\hat t_1\oplus\hat t_2 \in \hat T} 
		\chi_i(\hat s) 
		f_S(\hat s)
		f_T(\hat t)
		\tr (w(\hat s)  w(\hat t_1)) w(\hat t_2)  \\
		 &\propto
		 \sum_{\hat{t}_2 \in \hat U}
		 w(\hat{t}_2)
		 \Big(
			 \sum_{\hat{s} \in \hat{S}: \hat{s}\oplus \hat{u} \in T}
			 \chi_i(\hat{s}) f_S(\hat{s}) f_T(\hat{s} \oplus \hat{t}_2)
		\Big).
	\end{align*}
	The statement follows by invoking Lemma~\ref{lem:states in pauli span}.
\end{proof}

Next, we show a property of stabiliser code projectors used in the proof of Lemma \ref{lem:so-in-ad}.

\begin{lemma}
	\label{lem:non-diagonal-projector-stuff}
	Suppose that $P$ is a projector onto a stabiliser code.
	If $P$ is non-diagonal, then there are $\ket x \neq \ket y$ such that $P\ket x = P\ket y \neq 0$.
\end{lemma}

\begin{proof}
	We can assume that the stabiliser group of $P$ has $l\geq 1$ non-diagonal generators, given as $\omega^{s_j} w(z_j,x_j) $ for $s_j\in\F_d$, $ z_j,x_j \in \F_d^n $, and $x_j\neq 0$.
	The remaining, diagonal generators are of the form $\omega^{-z\cdot b} Z(z)$ for some $b\in\F_d^n$ and $z$ is an element of a suitable subspace $M$, such that
	\begin{align*}
		M \subset L:= \{ z \in \FF_d^n \; | \;  Z(z)w(z_j,x_j) = w(z_j,x_j) Z(z) \Leftrightarrow z \cdot x_j = 0 \quad \forall j=1,\dots,l \}.
	\end{align*}
	For any $\ket x$, its stabilisers are $\omega^{-z\cdot x}Z(z)$ for $z\in\F_d^n = L\oplus L^\perp$.
	Under the projection $P$, the stabilisers with $z\in L^\perp$ are replaced by the group generated by $\omega^{s_j} w(z_j,x_j)$.
	For $P\ket{x}$ to be non-zero it is then necessary and sufficient that $z\cdot x = z\cdot b$ for all $z\in M$.
	We then have $P\ket x = P \ket y \neq 0$  if moreover $z\cdot x = z\cdot y$ for all $z \in L$.
	Since this enforces $\dim L = n-l$ constraints on $x$, there are $d^l>1$ possible solutions.
    Thus, we can always find at least two distinct states $x\neq y$ such that $P\ket x = P \ket y \neq 0$, as claimed.
\end{proof}

%%% =============================================
\section{Properties of the counter-example \texorpdfstring{$\Lambda$}{Lambda}}
\label{sec:on-lambda}
%%% =============================================

In this section, we analyse the properties of the $ \Lambda $-channel in the case of qubits.\footnote{A similar analysis can also be done for qudits which is however more evolved.} 
Its action on the computational basis is given by
\begin{equation}
	\label{eq:lambda-channel}
	\Lambda(\rho) := \rho_{00} \ketbra{+}{+} +  \sum_{x\neq 0} \rho_{xx} \ketbra{x}{x} + \frac 1 2 \sum_{\substack{x\neq y \\ x \neq 0 \neq y}} \rho_{xy} \ketbra{x}{y}, \qquad \rho_{xy}:=\sandwich{x}{\rho}{y}.
\end{equation} 
From the definition, it is evident that $\Lambda$ is trace-preserving. 
However, it is not obvious that $\Lambda$ is completely stabiliser-preserving, a fact which is proven by Lem.~\ref{lem:lambda-extremality-AD}.
Here, we give an independent, self-contained proof for the CSP property which also sheds a bit of light on the interpretation of the channel $\Lambda$.

To this end, we claim that $\Lambda$ has a Kraus decomposition given by 
\begin{equation}
\label{eq:lambda-kraus-decomposition}
\Lambda(\rho) = H^{\otimes n}\ketbra{0}{0}\rho\ketbra{0}{0}H^{\otimes n} + \frac{1}{2^{n-1}}\sum_{z \in\F_2^n\setminus 0} P_z \rho P_z,
\end{equation}
where $P_z=(\one-Z(z))/2$ projects onto the stabiliser code given by the span of computational basis states $\ket{x}$ with $x\cdot z \neq 0$.
Then, by the polar decomposition, Eq.~\eqref{eq:csp-polar-decomposition}, of CSP channels discussed in App.~\ref{sec:additional}, the Kraus decomposition \eqref{eq:lambda-kraus-decomposition} defines a CSP channel since 
\begin{equation}
 \ketbra{0}{0} + \frac{1}{2^{n}}\sum_{z \in\F_2^n\setminus 0} \left(\one - Z(z)\right) = \one + \ketbra{0}{0} - \frac{1}{2^{n}}\sum_{z \in\F_2^n} Z(z) = \one.
\end{equation}
Alternatively, it is also straightforward to compute the Choi state from Eq.~\eqref{eq:lambda-kraus-decomposition}.
Let us define for any $z\in\F_2^n\setminus 0$ the affine subspace $K_z:=\{x\in\F_2^n: z\cdot x =1\}$ and the $2n$-qubit stabiliser state
\begin{equation}
 \ket{\psi_z} := 2^{- \frac{n-1}{2} }\sum_{x\in K_z} \ket{xx}.
\end{equation}
Then, the Choi state is
\begin{equation}
 \mathcal{J}(\Lambda) = \frac{1}{2^n} \left( \ketbra{+}{+}\otimes\ketbra{0}{0} + \sum_{z\neq 0} \ketbra{\psi_z}{\psi_z} \right),
\end{equation}
which lies in the stabiliser polytope $ \SP_{2n} $.

Finally, to prove the Kraus decomposition \eqref{eq:lambda-kraus-decomposition}, we check that it agrees with Eq.~\eqref{eq:lambda-channel} on the computational basis.
To this end, let us denote the channel Eq.~\eqref{eq:lambda-kraus-decomposition} as $\tilde\Lambda$.
Note that $P_z\ketbra{x}{y}P_z$ is zero if and only if $x$ or $y$ is orthogonal to $z$ and $ \ketbra{x}{y} $ otherwise.
Thus, $\tilde\Lambda(\ketbra{0}{0}) = \ketbra{+}{+}$.
For any $x\neq 0$, the linear equation $x\cdot z=1$ has exactly $2^{n-1}$ solutions $z\in\F_2^n$.
Since the first term in Eq.~\eqref{eq:lambda-kraus-decomposition} yields 0, we get $\tilde\Lambda(\ketbra{x}{x}) = \ketbra{x}{x}$ for any $x\neq 0$.
Furthermore, adding the condition $y\cdot z=1$ for any $y\notin\{ 0,x\}$ will further half the solution space, yielding $2^{n-2}$ vectors which are not orthogonal to both $x$ and $y$.
Thus, given two non-zero vectors $x\neq y$, we get $\tilde\Lambda(\ketbra{x}{y}) = \frac 1 2 \ketbra{x}{y}$ which then shows that $\tilde\Lambda=\Lambda$.

A natural question to ask is whether $\Lambda$ can be expressed in terms of more elementary quantum channels.
We can write the channel as a composition of the following three operations: 
\begin{enumerate}
 \item Perform a projective measurement with projectors $\{\ketbra{0^n}{0^n},\one-\ketbra{0^n}{0^n}\}$. This channel sets all off-diagonal terms in the first row and column of $\rho$ to zero, i.e.~it block-diagonalises $\rho$ with respect to the entry at position $ (0,0) $.
 \item Partial dephasing in the computational basis with probability $1/2$. This channel reduces the amplitude of the off-diagonal terms by $1/2$.
 \item Apply a global Hadamard gate on all qubits conditioned on the ``0'' outcome of the measurement.
\end{enumerate}
Interestingly, all three components are necessary for $\Lambda$ to have the desired properties.
If we leave out the second channel, it is possible to show that the composition of 1 and 3 is not stabiliser-preserving for $n\geq 2$\footnote{This can in principle be done by computing the Choi states of the corresponding channels and then finding a hyperplane that separates them from the stabiliser polytope $ \SP_{2n} $}, while for $n=1$ it is simply a stabiliser operation. 
Moreover, if we leave out channel 2 and 3, then we can rewrite the block-diagonalisation as a uniform convex combination of the identity and the diagonal $n$-qubit gate $V_n:=\diag(-1,1,\dots,1)$.
Note that $V_n = X^{\otimes n} (C^{n-1}Z) X^{\otimes n}$, thus it is in the $n$-th level of the Clifford hierarchy.
Hence, for $n\leq 2$, this is a mixed Clifford channel and in particular a stabiliser operation.
For $n>2$, the same technique as before can be used to show that this channel is not CSP. 
The effect of the dephasing channel is to sufficiently reduce the ``magic'' of the overall channel.
With increasing dephasing strength, it approaches the CSP polytope from the outside and eventually becomes CSP.
Figuratively speaking, the Hadamard gate in the last step fine-tunes the direction from which the CSP polytope is being approached, resulting in a channel which is a vertex.

%%% =============================================
\section{Measurements that are not followed by adaptive operations are never extremal}
\label{app:Measurements that are not followed by adaptive operations are never extremal}
%%% =============================================

\begin{lemma}\label{lem:NoContOperations-NoExtremality}
Suppose $ \mathcal{E} = \sum_{i} K_i \cdot K_i \in \CSP_n$ is an extremal CSP map with Kraus operators $ K_i $. Then $ \mathcal{E} $ does not contain a set of $ d $ Kraus operators that are Pauli measurements of the form $ PP_0,...,PP_{d-1} $ for some fixed Pauli projector $ P $ and where $P_0,...,P_{d-1} $ are projectors onto the $ d $ eigenspaces of some Pauli operator $ w(z,x) $. 
\end{lemma}
\begin{proof}
Let $ \mathcal{O} $ be the map that is composed of the $ d $ Kraus operators $ PP_0,...,PP_{d-1} $, i.e.
\begin{align*}
	\mathcal{O}(\rho) = \sum_{x = 0}^{d-1} PP_x \rho P_xP. 
	%			= P\Big  (\sum_{x = 0}^{d-1} P_x \rho P_x  \Big ) P .
\end{align*}
There is a Clifford $ C $ such that $ PP_x = C \tilde{P}  \otimes \ketbra{x}{x}  C^\dagger$ for a projector $ \tilde{P} $ acting on the first $ n-1 $ qudits. 
Define the operation 
\begin{align}
	\mathcal{M}(\rho ) =\sum_{x = 0}^{d-1}  \tilde{P}\otimes \ketbra{x}{x} \rho  \tilde{P} \otimes \ketbra{x}{x},
\end{align}
so $ \mathcal{O}(\rho) = C \,  \mathcal{M}(C^\dagger \rho C)    \, C^\dagger  $.

Let $ s_0,...,s_{d-1} $ be the eigenstates of some Pauli operator $ w(z,x) $ for $ z,x \in \FF_1 $ and let $ C_i =  \diag(d s_i) \in \Cl_{1}$ be the diagonal Clifford unitary with diagonal proportional to $ s_i $. 
We claim that 
\begin{align}\label{eq:decomposition-measurement-into-Clifford}
	{\mathcal{M}}(\rho) = \frac{1}{d}\sum_{i = 0}^{d-1} (\tilde{P} \otimes C_i )\rho (\tilde{P} \otimes C_i^\dagger).
\end{align}
It suffices to check the equation for inputs of the form $  A \otimes \ketbra{x}{y} $ for $ x,y \in \FF_d $ and a Hermitian matrix $ A $ acting on $ n-1 $ qubits.
We have $ {\mathcal{M}}(A \otimes \ketbra{x}{y}) = \braket{x}{y} \tilde{P}A\tilde{P} \otimes \ketbra{x}{y} $
and for the RHS of \eqref{eq:decomposition-measurement-into-Clifford}
\begin{align}
	\frac{1}{d} \sum_{i = 0}^{d-1} \tilde{P}A\tilde{P} \otimes C_i \ketbra{x}{y}  C_i^\dagger = \frac{1}{d}\tilde{P}A\tilde{P} \otimes \sum_{i = 0}^{d-1} C_i \ketbra{x}{y} C_i^\dagger  
	&= \tilde{P}A\tilde{P} \otimes \sum_{i = 0}^{d-1} s_i(x) \overline{s_i(y)} \ketbra{x}{y}  \\
	&= \tilde{P}A\tilde{P}  \otimes \ketbra{x}{y}  \sum_{i = 0}^{d-1} s_i(x) \overline{s_i(y)}   \\
	&=\braket{x}{y}  \tilde{P}A\tilde{P}  \otimes \ketbra{x}{y},
\end{align}
where the last equality stems from the fact that 
\begin{align}
	\sum_{i = 0}^{d-1} s_i(x) \overline{s_i(y)}  = \Big (\sum_{i = 0}^{d-1} \ketbra{s_i}{s_i} \Big )(x,y) = \one_{d}(x,y) = \braket{x}{y} \ketbra{x}{y}. 
\end{align}
Hence, 
\begin{align}
	\mathcal{O}(\rho)  = C  \, \mathcal{M}(C \rho C^\dagger)  \, C^\dagger  &= C {M}(C^\dagger \rho C)  \, C^\dagger    \\
	&= \frac{1}{d}  \sum_{i = 0}^{d-1}  C (\tilde{P} \otimes C_i )C^\dagger \, \rho \, C (\tilde{P} \otimes C_i) C^\dagger .
\end{align}
If the original channel $ \mathcal{E} $ decomposes as $ \mathcal{E} = \mathcal{O} + \mathcal{O}^c $, then we can write it now as a convex combination of distinct operations 
\begin{align*}
	\mathcal{E} =\frac{1}{d} (\mathcal{E}_1 + \dots +\mathcal{E}_{d-1}  )\quad \text{with}  \quad 
	\mathcal{E}_i(\rho) =  C (\tilde{P} \otimes C_i )C^\dagger \, \rho \, C (\tilde{P} \otimes C_i) C^\dagger  + \mathcal{O}^c(\rho).
\end{align*}
The maps $ \mathcal{E}_i $ are completely positive and trace-preserving because 
\begin{align*}
	\mathcal{E}_i^\dagger(\one_n) &= C (\tilde{P} \otimes C_i^\dagger) C^\dagger \one_{n} C (\tilde{P} \otimes C_i)  C^\dagger  + (\mathcal{O}^c)^\dagger(\one_n)      \\
	&= C (\tilde{P} \otimes \one_1)  C^\dagger   + (\mathcal{O}^c)^\dagger(\one_n)   \\
	&= \mathcal{O}^\dagger(\one_n)  + (\mathcal{O}^c)^\dagger(\one_n)   \\
	&= \mathcal{E}^\dagger(\one_n)  \\
	&= \one_n,
\end{align*}
where the third equation follows from
\begin{align*}
	\mathcal{O}^\dagger(\one_n) &= \sum_{x = 0}^{d-1} C (\tilde{P} \otimes \ketbra{x}{x}) C^\dagger \one_n C (\tilde{P} \otimes \ketbra{x}{x})  C^\dagger   \\
	&= C \Big (\sum_{x = 0}^{d-1} \tilde{P} \otimes \ketbra{x}{x}  \Big ) C^\dagger   \\
	&= C (\tilde{P} \otimes \one_1 )C^\dagger.
\end{align*}
This proves that $ \mathcal{E} $ cannot be extremal. 
\end{proof}

\end{document}